\documentclass[12pt]{article}
\usepackage[english]{babel}
\usepackage[letterpaper,top=2.5cm,bottom=2.5cm,left=2.5cm,right=2.5cm,marginparwidth=1.75cm]{geometry}
\usepackage{amsmath,amsthm,amssymb,enumitem,bbm}
\usepackage{graphicx}
\usepackage{bm}
\usepackage[colorlinks = true,linkcolor = blue,urlcolor  = blue,citecolor = blue,anchorcolor = blue,hyperfootnotes=false]{hyperref}
\usepackage[symbol]{footmisc}
\usepackage{booktabs}

\usepackage[giveninits=true,backend=biber, maxnames=4,maxbibnames=99,maxalphanames=4,style=alphabetic,doi=true,isbn=true,url=false,eprint=false]{biblatex}
\renewbibmacro{in:}{}
\AtEveryBibitem{%
  \clearlist{language}%
  \clearfield{note}%
  \clearfield{issn}%
}
\AtBeginBibliography{\emergencystretch=3em}
\DeclareSourcemap{%
  \maps[datatype=bibtex]{%
    \map{%
      \step[fieldsource=entrykey, match=\regexp{sp_global_market_intelligence_credit_2024}, final]%
      \step[fieldset=shorthand, fieldvalue={SPG24}]%
    }%
    \map[overwrite]{\step[fieldset=abstract, null]}%
    \map[overwrite]{\step[fieldset=rights, null]}%
    \map[overwrite]{\step[fieldset=note, null]}%
    \map[overwrite]{
      \step[fieldsource=shortjournal, final]%
      \step[fieldset=journaltitle, origfieldval]%
    }%
  }%
}

\newcommand{\R}{\mathbb{R}}
\newcommand{\calF}{\mathcal{F}}

\newcommand{\lb}{\left}
\newcommand{\rb}{\right}
\newcommand{\ind}{\mathbbm{1}}
\newcommand{\real}{\operatorname{Re}}
\newcommand{\wt}{\widetilde}
\newcommand{\bQ}{\mathbb{Q}}
\newcommand{\bF}{\mathbb{F}}
\newcommand{\E}{\mathbb{E}}
\newcommand{\Exp}{\operatorname{Exp}}

\theoremstyle{plain}
\newtheorem{theorem}{Theorem}[section]
\newtheorem{lemma}[theorem]{Lemma}

\theoremstyle{remark}
\newtheorem{remark}[theorem]{Remark}

\allowdisplaybreaks[3]

\begin{document}

\begin{center}
\Large
Multi-Credit Calibration via\\Elastically Stopped L\'{e}vy Processes

\normalsize

\vspace{1em}

Graeme Baker\footnote{Department of Statistics, Columbia University, NY, USA \href{mailto:g.baker@columbia.edu}{g.baker@columbia.edu}} and Agostino Capponi\footnote{Department of Industrial Engineering and Operations Research, Columbia University, NY, USA \href{mailto:ac3827@columbia.edu}{ac3827@columbia.edu}} 

\vspace{2em}

\end{center}

\begin{abstract}
We calibrate credit default swaps and index tranches with elastically stopped L\'evy processes: each firm defaults when the running supremum of a latent, spectrally positive distress process crosses an independent exponential barrier. This yields a Cox construction with totally inaccessible default times, while retaining the interpretability and explicit formulas of a structural approach. Adding a single common compound Poisson jump factor to every firm's latent driver gives a parsimonious multi-credit model with simultaneous defaults, which is priced by an exact Wiener--Hopf Monte Carlo scheme. Its tractability rests on a single-name result we prove: a finite partial-fraction formula for the Laplace transform of the default probability under phase-type jumps. On daily CDX North American High-Yield and Investment-Grade panels, our drivers attain the lowest out-of-sample errors in a six-model field and reproduce the inverted spread curves of names heading into default, which a L\'evy subordinator provably does not. At the index level, the two-parameter dependence structure closes $73\%$ to $89\%$ of the tranche pricing gap left by independent marginals with the dependence parameters frozen, and up to $95\%$ once re-marked to tranche quotes; our framework dominates a single-factor Gaussian copula and the affine intensity benchmark of Duffie--G\^arleanu on both indices.
\end{abstract}

\section{Introduction}
Credit-risk models are typically classified in one of two forms: \emph{structural approaches}, in which each firm defaults when its asset value falls below its liabilities, and \emph{intensity-based approaches}, in which the default time is specified exogenously. The most common form of an intensity-based approach is the Cox construction, in which the default time is the first time an increasing process reaches an independent exponentially distributed random variable. What distinguishes it from the structural approach is that the process being tracked is latent, rather than linked to balance sheet fundamentals. Herein, we bridge the two approaches by driving Cox constructions with the running suprema of latent L\'{e}vy processes, which recasts the default times as \emph{elastic stopping times} at the suprema. This framing brings with it the well-developed first-passage theory of L\'{e}vy processes, yielding explicit formulas for the Laplace transform of credit default swap (CDS) term structures in the case of spectrally positive L\'{e}vy processes, as well as a parsimonious and tractable framework for joint defaults and multi-credit calibration.

We fix a filtered probability space $(\Omega,\calF,\bF=(\calF_t)_{t\ge0},\bQ)$ under a risk-neutral measure and consider $n\ge 1$ firms at risk of default. For $1\le i\le n$, we model the default time $\tau_i$ by
\begin{equation}\label{eq:taudef}
\tau_i=\inf\lb\{t\ge 0: \Lambda^i_t\ge \xi_i\rb\}
\end{equation}
where $\xi_1,\dots,\xi_n$ are mutually independent $\Exp(1)$ random variables which are also independent of $\calF_\infty$, and $\Lambda^i$ is a non-negative, non-decreasing right-continuous process. In this work, we take
\begin{equation}\label{eq:Lambdadef}
\Lambda^i_t = \sup_{s\le t} \lb(X_s^i\rb)_+,
\qquad\text{where } (z)_+:=\max(0,z) \text{ for } z\in\R,
\end{equation}
and $X^i$ is a latent distress process. Since $z\mapsto (z)_+$ is non-decreasing, $\Lambda^i$ is equivalently given as the positive part of the running supremum of $X^i$, and it is thus a genuine cumulative-hazard process, non-negative and non-decreasing, with $\Lambda_0^i=(X_0^i)_+$. The potential tractability of this model boils down to the computability of the default probabilities. Note first that $\bQ\lb(\tau_i\le T\rb)=\bQ\lb(\Lambda_T^i\ge \xi_i\rb)$. Iterated conditioning and the cumulative distribution function (CDF) of an $\Exp(1)$ random variable yield
\begin{equation}\label{eq:pd}
\bQ\lb(\tau_i\le T\rb)=\E^\bQ\lb[\E^\bQ\lb[\ind_{\{\Lambda_T^i\ge \xi_i\}}\mid \calF_T\rb]\rb]=\E^\bQ\lb[1-e^{-\Lambda_T^i}\rb].
\end{equation}
The goal of this article is to calibrate \eqref{eq:pd} to CDS data by choosing $X$ from a suitable class of L\'{e}vy processes. Our destination is a multi-credit model for pricing index tranches, and along the way we develop the single-name theory that makes the model both exact to price and cheap to simulate. We put forward three main contributions:
\begin{enumerate}[leftmargin=2em]
\item \emph{An exact pricing formula for CDS using elastically stopped spectrally positive L\'{e}vy processes.} The drivers we calibrate (see \eqref{eq:driver} below) have no negative jumps, rendering \eqref{eq:pd} computable through the fluctuation theory of spectrally positive L\'evy processes (see, for instance, \cite{kyprianou_fluctuations_2014,cohen_theory_2012}). For any jump law in the phase-type class (which contains the exponential and Erlang laws, is closed under sums and mixtures, and is dense in the laws on $[0,\infty)$), Theorem~\ref{thm} gives the Laplace transform of the default probability as a finite partial-fraction formula in the roots of a single polynomial. In practice, the Laplace transform is inverted with a Talbot contour (see Appendix~\ref{sec:Talbot}). We will see that the roots of the target polynomial depend only on a set of static parameters, so re-marking a firm by its initial distress level amounts to evaluation of a single vector of complex exponentials.
\item \emph{A joint model and an exact simulation method for multi-credit pricing.} Adding a single common compound Poisson jump factor to every firm's latent driver produces simultaneous defaults while preserving phase-type marginals (Section~\ref{sec:common-wh}). The joint law is priced by a Wiener--Hopf Monte Carlo scheme whose exponential clock is perfectly suited for the common factor's arrivals, and whose sampling distributions are assembled from the same polynomial roots as the pricing formula of the previous item, so the joint model runs on the machinery already built for the marginals. The scheme carries no time-discretization bias, and it reaches the accuracy of a fine-grid Euler simulation well over an order of magnitude faster.
\item \emph{Empirical validation against structural, affine, and copula benchmarks on two indices.} On daily constituent panels of the CDX North American High-Yield (NAHY) and Investment-Grade (NAIG) indices, the spectrally positive L\'{e}vy drivers attain the lowest in-sample and out-of-sample errors in a six-model field. Furthermore, they reproduce the inverted spread curves of names heading into default, which are empirically more difficult to capture with an affine intensity model, and provably cannot be captured by a L\'{e}vy subordinator. At the index level, the common-jump construction closes $73\%$ to $89\%$ of the tranche pricing gap left by independent marginals with the dependence parameters frozen, and up to $95\%$ with the dependence parameters re-marked. Our method dominates a one-factor Gaussian copula benchmark \cite{li_default_2000} and an affine intensity benchmark \cite{duffie_risk_2001} on both indices.
\end{enumerate}

A shortcoming of the structural approach for pricing credit derivatives is the predictability of the default time. In the Black--Cox model \cite{black_valuing_1976} the default time is the first passage of a continuous process below a barrier, and is therefore an \emph{accessible} stopping time: it is announced by the approach of the process to the barrier, so a firm which is not currently near default cannot default in the next instant (see \cite[Chapter III]{protter_stochastic_2005}). The consequence is visible in prices. Short-dated spreads implied by a Black--Cox model vanish as the maturity shrinks (see \cite[Subsection 3.1.5]{brigo_counterparty_2013}), and our own results corroborate this, with Black--Cox producing the largest in-sample calibration error on the investment-grade index.

Two further repairs are available, and we find that they work best in combination. The first is to randomize the barrier that triggers default. Replacing the fixed level of the structural model by the independent $\Exp(1)$ barrier in equation \eqref{eq:taudef} makes $\tau_i$ totally inaccessible even when $X^i$ has continuous paths. We will also refer to the default mechanism of $\tau_i$ as \emph{elastic stopping} of $X^i$ at its supremum: the running supremum increases only at the times when the driver reaches a new maximum, and the process is therefore killed at a rate proportional to the local time the driver spends at its running supremum. The second repair is to give the latent process jumps, which aligns with the modelling intuition that distress arrives as a sequence of sudden adverse shocks and moves gradually in between. Combining the two, the model to keep in mind is elastic stopping of the driver
\begin{equation}\label{eq:driver}
X^i_t = -\delta_i + \mu t + \sigma B_t + \sum_{k\le N_t} U_k,
\qquad U_k\ \text{i.i.d.}\ \Exp(\eta),
\end{equation}
where $\mu\in\R$, $\sigma\ge0$, $\eta>0$, $N$ is a Poisson process of rate $\lambda_J>0$ independent of the Brownian motion $B$, and $\delta_i\ge0$ is an initial \emph{distance-to-distress}. Notice that we keep the interpretable latent process of the structural approach, which the intensity-based approach gives up, while obtaining the totally inaccessible default time which the structural approach cannot produce. For $\sigma>0$ we call this driver LBM+Exp, for linear Brownian motion with exponential jumps; our main theorem handles the general class of phase-type jumps.

Equation \eqref{eq:driver} nests the special cases we consider in the paper. Taking $\sigma=0$ and $\mu<0$ recovers the classical Cram\'er--Lundberg risk process of insurance mathematics, stopped at a randomized upper barrier rather than ruin at zero, and taking $\lambda_J=0$ recovers the linear Brownian motion of Black--Cox, so that \eqref{eq:driver} adds together the foundational first-passage processes of insurance and of finance. Lastly, taking $\sigma=0$ and $\mu\ge0$ makes $X^i$ non-decreasing, a L\'evy \emph{subordinator}, which is the driver class already used elsewhere in the multi-credit literature (see, for instance, \cite{bo_counterparty_2015,sun_marshallolkin_2017}). The distinction to have in mind between these subcases is one of chronic versus acute distress. A subordinator only ever climbs, so its accumulated hazard rises to a plateau and can never decline; it models a firm which only ever deteriorates. A driver with $\mu<0$ instead recovers between shocks, so a firm surviving a bout of acute distress has improved prospects afterwards, and its hazard can decay. The distinction is falsifiable, since a subordinator can never produce an inverted CDS spread curve while the $\mu<0$ class can, and Section~\ref{sec:onesided} makes both claims precise. The data itself favours the negative-drift class, with most high-yield firms fitting $\mu<0$ when the drift is left free (Subsection~\ref{sec:inversion}).

We calibrate and rank this driver family, together with Black--Cox and a Cox--Ingersoll--Ross intensity model with and without jumps (CIR and JCIR, as in \cite{brigo_counterparty_2013}), against daily CDS quotes on both indices in a rolling-window design. LBM+Exp attains the lowest median error on both indices, Cram\'er--Lundberg is close behind with one fewer parameter, and the subordinator trails both clearly. The sharper test is the falsifiable one above. Genuine spread-curve inversions are rare, concentrated in distressed CCC names, and disproportionately precede bankruptcy. Out of sample, our negative-drift drivers reproduce over $80\%$ of them, using the state $\delta_i$ as a shape parameter rather than by fitting the curve's slope directly; the subordinator reproduces none, exactly as the theory predicts; and the affine benchmarks reproduce them far less consistently (Subsection~\ref{sec:inversion}).

Turning to the multi-credit picture, our goal is to price index tranches while staying faithful to the single-name marginals. Summing each tranche's absolute error as a fraction of its notional across the four tranches, independent marginals miss a year of tranche quotes by $0.40$ to $0.44$ per board (see Appendix~\ref{sec:tranche-calibration} for the pricing conventions). We close most of that gap by adding a single common compound Poisson shock onto every firm's driver. At each arrival of the shock, every firm's distress jumps by the same amount, so that any number of firms may default simultaneously. Two observations make the construction inexpensive. Seen from one firm, the driver accumulates two independent exponential jump streams, so the marginal is again phase-type and remains within reach of our single-name theory. Seen from the simulator, the common shock arrives at exponentially distributed times, which is precisely the time randomization on which the Wiener--Hopf scheme is built. The dependence requires only two extra parameters, and out of sample it closes $73\%$ to $89\%$ of the independence gap with the dependence parameters frozen, and up to $95\%$ fully re-marked, dominating both a one-factor Gaussian copula \cite{li_default_2000} and an affine intensity benchmark \cite{duffie_risk_2001} (6 common parameters) on both indices. The structural separation among the three subcases of \eqref{eq:driver} persists at this level, much as it does for single names.

The rest of the paper is organized as follows. After a review of the literature, Section~\ref{sec:onesided} develops the computation of \eqref{eq:pd} and the structural separation inherent in \eqref{eq:driver}. In Section~\ref{sec:common-wh}, we build the common-jump joint model and its exact Wiener--Hopf Monte Carlo scheme, and benchmark the scheme's speed. Section~\ref{sec:data} describes the data and the conventions used downstream. Sections~\ref{sec:calibration} and \ref{sec:joint-calibration} calibrate the marginals and the joint law respectively, and Section~\ref{sec:conclusion} concludes.

\subsection{Literature Review}\label{sec:litreview}

First-passage modelling of solvency dates back to Lundberg's insurance risk model from 1903 \cite{lundberg_approximerad_1903}: an insurer's reserves grow at a constant rate as premiums accrue, fall by a jump at each claim, and ruin is the first passage below zero. The negative of Lundberg's process is exactly the distress process of \eqref{eq:driver} at $\sigma=0$: distress drifts downward as the firm services its obligations and jumps upward when a shock lands, and default occurs when distress sets a sufficiently high record. 
Cram\'er extended and popularized this theory \cite{cramer_historical_1994}, and the resulting Cram\'er--Lundberg process is an instance of the general class of processes with stationary and independent increments, the L\'evy processes we use throughout.

The two classical strands of credit modelling are by now standard, and we recall them only to place our present construction. Structural models begin with \cite{merton_pricing_1974}, where default occurs if the value of the firm's assets falls short of its liabilities at a fixed maturity, and with \cite{black_valuing_1976}, where default is instead the first passage of the asset value below a barrier at any time before maturity. \cite{leland_optimal_1996} make the barrier itself endogenous, choosing it optimally together with the maturity structure of the debt, and obtain term structures of credit spreads which match more closely with historical averages. Despite their economic appeal, structural credit risk models have seen limited use in pricing credit-sensitive products, since in their basic forms they struggle to fit the observed term structure of credit spreads. Intensity-based models take the opposite starting point and specify the law of the default time directly. The Cox construction underlying \eqref{eq:taudef} is due in this context to \cite{lando_cox_1998}, where the intensity is driven by an exogenous state process. \cite{duffie_multi-period_2007} demonstrate the empirical reach of that approach for multi-period default prediction, driving the intensity with firm-level covariates that include the distance to default of the structural lineage. For further background, we refer the reader to the textbooks \cite{bielecki_credit_2004} and \cite{brigo_counterparty_2013}.

The two credit methodologies come together when a structural mechanism generates an intensity because of incomplete information on balance sheet fundamentals. \cite{duffie_term_2001} suppose that bondholders observe the firm's assets only through periodic and imperfect accounting reports, and show that the resulting filtration admits a default intensity even though the underlying first-passage time does not. \cite{cetin_modeling_2004} obtain a related reduction under partial information about the firm's cash flows. In those models, the asset process is imperfectly observed, and the intensity arises from the coarsening of the filtration; for us, the distress process is a latent state rather than a balance-sheet quantity, yet the model's own filtration carries it in full, and total inaccessibility comes from the independent exponential threshold rather than from discarded information.

The randomized threshold in \eqref{eq:taudef} is an instance of elastic stopping, which 
has found recent use in systemic-risk and mean-field contexts; see \cite{baker_zero_2022,hambly_mckean-vlasov_2022,hambly_spde_2025}.
These works show that the elastic formulation restores regularity over a Black--Cox style hitting-time approach.
The hitting-time formulation itself, in which a large number of firms interact endogenously through the times at which they default, is developed in \cite{nadtochiy_particle_2019,delarue_particle_2015}, with associated free-boundary and SPDE limits in \cite{hambly_spde_2019,baker_zero_2022,baker_singular_2024} and the common-noise case in \cite{hambly_contagious_2025}. We do not pursue a mean-field limit here, and the common-factor construction of Section~\ref{sec:common-wh} is a finite-dimensional object; the connection to our work is that the device which regularizes those problems, stopping against an independent exponential level in place of a fixed boundary, is the same device which gives our default times their short-end behaviour.

L\'evy processes have been used as credit drivers before, usually as subordinators, which we retain as a benchmark. \cite{bo_counterparty_2015} model default clustering for counterparty risk in CDS by driving the cumulative hazards with a subordinator, and \cite{sun_marshallolkin_2017} show how Marshall--Olkin distributions arise from subordinators and give efficient simulation schemes for the resulting credit portfolios. Multivariate subordination of Markov processes, with credit among its applications, is developed in \cite{mendoza-arriaga_multivariate_2016} and \cite{mai_tractable_2009}. 

Our benchmarks for the joint law come from the portfolio credit literature. The one-factor Gaussian copula of \cite{li_default_2000} became the market standard for tranche pricing, and its limitations are well documented (see the discussion in \cite[Section 3.4]{brigo_counterparty_2013}). A single correlation cannot price the whole tranche stack at once, so market practice is to fit a different correlation to each attachment point (i.e.~\emph{base correlation skew}), yielding four separate models instead of one joint model. \cite{duffie_risk_2001} adopt an intensity-based approach, giving each firm a linear combination of an idiosyncratic and a common JCIR process. We take \cite{li_default_2000} and \cite{duffie_risk_2001} as our two benchmarks in Section~\ref{sec:joint-calibration}, as they are respectively the industry and academic standards.

Our core mathematical machinery relies on the fluctuation theory of L\'evy processes. \cite{kyprianou_fluctuations_2014} is an excellent introductory reference, and more details on the scale functions we use in Section~\ref{sec:onesided} are collected in \cite{cohen_theory_2012}. The phase-type and matrix-exponential classes of jump law, which render those scale functions explicit, are treated in \cite{asmussen_applied_2003} and \cite{bladt_matrix-exponential_2017}. \cite{pistorius_maxima_2006} obtains the law of the maximum for a dense class of L\'evy processes with rational transforms, and \cite{lewis_wiener-hopf_2008} give the Wiener--Hopf factorization when the positive jumps have rational transform, which is the case we exploit.  \cite{egami_phase-type_2014} fit scale functions by phase-type approximation, which is another clever use of the same structure, but would lead to over-parametrization in our case. For simulation, we follow the Wiener--Hopf Monte Carlo approach formulated in \cite{kuznetsov_wienerhopf_2011}, extended to path functionals in \cite{ferreiro-castilla_applying_2015}, and with explicit factorizations for particular families collected in \cite{kuznetsov_wienerhopf_2010}. Other L\'evy classes (stable, $\beta$-class, symmetric) admit related representations of the supremum law through self-similarity or the Wiener--Hopf factorization \cite{kuznetsov_wienerhopf_2010,kwasnicki_suprema_2013}, but proved either too restrictive or too costly to calibrate. Structural first-passage models such as Black--Cox and Duffie--Lando \cite{duffie_term_2001} can be embedded in \eqref{eq:taudef}--\eqref{eq:Lambdadef} only at the cost of breaking either the L\'evy property or the independence of $\xi_i$, which is why Black--Cox enters our field in its native formulation.

Finally, it is worth recalling where the running suprema of L\'{e}vy processes have occurred elsewhere in mathematical finance. \cite{asmussen_russian_2004} price Russian and American puts under exponential phase-type L\'evy models, the supremum entering through the optimal stopping boundary; \cite{bernyk_predicting_2011} study the prediction of the ultimate supremum of a stable process with no negative jumps; and \cite{kyprianou_distributional_2007} give a distributional study of de Finetti's dividend problem for a general L\'evy insurance risk process. We are not aware of previous work in which the running supremum of a \emph{non-subordinator} L\'evy process serves as the cumulative hazard in a Cox construction calibrated to market data.

\section{Default Probabilities for Spectrally Positive Drivers}\label{sec:onesided}

The running suprema of L\'{e}vy processes are well-studied objects from the fluctuation theory of L\'evy processes. Recall that if $X^i$ is a real-valued L\'evy process, its \emph{L\'evy--Khintchine exponent} $\Psi_i$ is defined by
\[
\E^\bQ\lb[e^{i\theta X^i_t}\rb] = e^{-t\Psi_i(\theta)}, \qquad \theta\in\R.
\]
The L\'evy--Khintchine theorem (see, e.g., \cite[Theorem 1.3]{kyprianou_fluctuations_2014}) decomposes any such exponent into its \emph{triplet} $(\mu, \sigma^2, \nu)$ form,
\[
\Psi_i(\theta)
=
-i\mu\theta + \tfrac{1}{2}\sigma^2\theta^2
+ \int_{\R\setminus\{0\}}\lb(1 - e^{i\theta x} + i\theta x\ind_{\{|x|<1\}}\rb)\,\nu(dx),
\]
where $\mu\in\R$ is a drift, $\sigma^2\ge 0$ is the variance of the continuous-martingale part, and $\nu$ is a $\sigma$-finite L\'evy measure on $\R\setminus\{0\}$ with $\int (1\wedge x^2)\,\nu(dx)<\infty$. The process $X^i$ admits the path-level decomposition
\[
X^i_t = \mu t + \sigma B_t + \int_{|x|<1}\!\!\!x\,\bigl(N(t,dx) - t\,\nu(dx)\bigr) + \int_{|x|\ge 1}\!\!\!x\,N(t,dx),
\]
where $B$ is a standard Brownian motion and $N$ is an independent Poisson random measure on $[0,\infty)\times\R\setminus\{0\}$ with intensity $dt\otimes \nu(dx)$. We see that $X^i$ is the sum of a drift, a Brownian component, a compensated ``small jump'' martingale, and a ``big jump'' compound Poisson process with rate $\nu(\{|x|\ge1\})$. The jump components give the class a huge range of flexibility for fitting distributions beyond Gaussian (see, for instance, the use of CGMY L\'{e}vy processes to fit equity returns in \cite{carr_fine_2002}).
We will also add an initial condition $X_0^i=x_0^i\neq 0$ to the above sum, and we will be explicit when we do so. 

We pursue here the computation of \eqref{eq:pd} using spectrally positive L\'evy processes. Given an initial condition $x_0^i$, we write
\[
X_t^i=x_0^i+\wt X_t^i,\qquad t\ge 0,
\]
where $\wt X^i$ is a spectrally positive L\'evy process with $\wt X_0^i=0$. Outside this subsection, we will write $X_0^i=-\delta$ with $\delta\ge 0$ when the initial condition is non-positive. This reparameterization is already implicit in \eqref{eq:driver}. It puts the driver on the safe side of the exponential barrier, so that there is no probability of instantaneous default (equivalently, the barrier for $\wt X^i$ is shifted up by $\delta$). We call $\delta$ the \emph{distance-to-distress}.

Throughout the paper, we follow the convention that subordinators, i.e.~non-decreasing L\'{e}vy processes, are not spectrally positive. Note that we have included a truncation by $(\cdot)_+$ in \eqref{eq:Lambdadef} so that $\Lambda^i$ gives a non-negative hazard process, but this can be handled with careful bookkeeping. The running supremum of $X^i$ is given by
\[
\sup_{0\le s\le t}X_s^i
=x_0^i+\wt\Lambda_t^i,
\qquad
\wt\Lambda_t^i:=\sup_{0\le s\le t}\wt X_s^i,
\]
so that $\Lambda^i_t=\lb(x_0^i+\wt\Lambda^i_t\rb)_+$ by \eqref{eq:Lambdadef}. Since $\Lambda^i$ is non-negative, the identity
 \[
1-e^{-z}
=
\int_0^\infty e^{-y}\ind_{\{z>y\}}\,dy,
\qquad z\ge 0
\]
yields
\begin{equation}\label{eq:FT}
\bQ(\tau_i\le T)
=
\E^\bQ\!\lb[1-e^{-\Lambda_T^i}\rb]
=
\int_0^\infty e^{-y}\,
\bQ\!\lb(\Lambda_T^i>y\rb)\,dy,
\end{equation}
and for $y\ge0$ we have $\{\Lambda^i_T>y\}=\{\sup_{s\le T}X^i_s>y\}$. The running supremum $\wt\Lambda^i$ is the quantity studied in the literature, and so we continue to keep track of the shift $x_0^i$ below.

\subsection{Laplace Transform of the Default Probability}\label{sec:fluctuation}

\cite[Chapter 8]{kyprianou_fluctuations_2014} gives a concise introduction to fluctuations of L\'{e}vy processes in the spectrally negative case, and we apply these results to the running infimum of the dual spectrally negative process $-X^i$. Define the upward passage time of $X^i$ above level $x$ by
\[
\sigma_{i,x}^+ := \inf\{t\ge0:X_t^i>x\},\qquad x\ge0.
\]
Note that
\[
\{\Lambda_T^i>x\}=\{\sigma_{i,x}^+\le T\}.
\]
Write
\[
\psi_i(\theta)=\log \E^\bQ\lb[e^{-\theta X_1^i}\rb],\qquad \theta\ge0,
\]
for the Laplace exponent (technically, this is the Laplace exponent of $-X^i$) and let $\beta_1(q)$ be the largest non-negative solution of $\psi_i(\theta)=q$.

Next, we consider the $q$-\emph{scale functions}, $W_i^{(q)}$ and $Z_i^{(q)}$. $W_i^{(q)}$ is the unique function with $W_i^{(q)}(x)=0$ for $x<0$ and
\[
\int_0^\infty e^{-\beta x}W_i^{(q)}(x)\,dx
=
\frac{1}{\psi_i(\beta)-q},
\qquad \beta>\beta_1(q).
\]
We define $Z_i^{(q)}$ by
\[
Z_i^{(q)}(x)=1+q\int_0^x W_i^{(q)}(y)\,dy,\qquad x\in\R,
\]
so that $Z_i^{(q)}(x)=1$ for $x<0$, since $W_i^{(q)}$ vanishes there.

We present our first result:
\begin{lemma}\label{lem:0init}
Suppose that $\wt X^i$ is a spectrally positive L\'{e}vy process with $\wt X_0^i=0$. Then for $q>0$, we have that
\begin{align*}
\int_0^\infty e^{-qT}\,\bQ(\wt\Lambda_T^i>x)\,dT
&=
\frac{1}{q}
\lb(
Z_i^{(q)}(x)-\frac{q}{\beta_1(q)}W_i^{(q)}(x)
\rb).
\end{align*}
\end{lemma}

\begin{proof}
\cite[Theorem 8.1]{kyprianou_fluctuations_2014} implies
\[
\E^\bQ\!\lb[e^{-q\sigma_{i,x}^+}\ind_{\{\sigma_{i,x}^+<\infty\}}\rb]
=
Z_i^{(q)}(x)-\frac{q}{\beta_1(q)}W_i^{(q)}(x),
\qquad q>0,\ x\ge0.
\]
Applying the Fubini--Tonelli Theorem, we obtain
\begin{align*}
\int_0^\infty e^{-qT}\,\bQ(\wt\Lambda_T^i>x)\,dT
&=
\int_0^\infty e^{-qT}\,\E^\bQ\!\lb[\ind_{\{\wt\Lambda_T^i>x\}}\rb]\,dT \\
&=
\int_0^\infty e^{-qT}\,\E^\bQ\!\lb[\ind_{\{\sigma_{i,x}^+\le T\}}\rb]\,dT \\
&=
\E^\bQ\!\lb[\ind_{\{\sigma_{i,x}^+<\infty\}}\int_{\sigma_{i,x}^+}^\infty e^{-qT}\,dT\rb] \\
&=
\frac{1}{q}\,
\E^\bQ\!\lb[e^{-q\sigma_{i,x}^+}\ind_{\{\sigma_{i,x}^+<\infty\}}\rb]\\
&=
\frac{1}{q}
\lb(
Z_i^{(q)}(x)-\frac{q}{\beta_1(q)}W_i^{(q)}(x)
\rb).\qedhere
\end{align*}
\end{proof}

With the next lemma, we incorporate the initial condition $X_0^i=x_0^i\in\R$. In Subsection~\ref{subsec:panel}, we treat $x_0^i$ as a one-dimensional state that marks the firm's current distress level, while holding the remaining L\'evy-triplet parameters static across dates.

\begin{lemma}\label{lem:laplace}
Suppose that $X^i$ is a spectrally positive L\'{e}vy process and let $X_0^i=x_0^i\in\R$. Then for $q>0$, we have that
\[
\int_0^\infty e^{-qT}\,\bQ(\tau_i\le T)\,dT
\;=\;
\int_0^\infty \frac{e^{-y}}{q}\lb[Z_i^{(q)}(y-x_0^i) - \frac{q}{\beta_1(q)}\,W_i^{(q)}(y-x_0^i)\rb]\,dy.
\]
\end{lemma}

\begin{proof}
Equation \eqref{eq:FT} yields
\[
\bQ(\tau_i\le T) \;=\; \int_0^\infty e^{-y}\,\bQ\lb(x_0^i+\wt\Lambda_T^i> y\rb)\,dy \;=\; \int_0^\infty e^{-y}\,\bQ\lb(\wt\Lambda_T^i> y-x_0^i\rb)\,dy.
\]
Hence, Fubini--Tonelli once more gives
\[
\int_0^\infty e^{-qT}\,\bQ(\tau_i\le T)\,dT
\;=\;
\int_0^\infty e^{-y}
\lb(
\int_0^\infty e^{-qT}
\bQ\!\lb(
\wt\Lambda_T^i>y-x_0^i
\rb)\,dT
\rb)\,dy.
\]
If $y-x_0^i<0$, then $\wt\Lambda_T^i\ge0>y-x_0^i$ for all $T$, and the inner integral equals $1/q$. If $y-x_0^i>0$, Lemma~\ref{lem:0init} applies. The boundary case $y=x_0^i$ affects the outer integral on a Lebesgue-null set only. Since $W_i^{(q)}(x)=0$ and $Z_i^{(q)}(x)=1$ for $x<0$, we obtain the desired identity.
\end{proof}

\subsection{Phase-Type Jumps}
We now specialize the jump law to the phase-type class (see \cite[Chapter III]{asmussen_applied_2003} and \cite{bladt_matrix-exponential_2017} for background). This class is dense in the space of laws on $[0,\infty)$ under the topology of weak convergence, and closed under mixtures.

A phase-type law $U\sim\mathrm{PH}(\bm\alpha, \mathbf{T})$ on $m$ phases is the absorption time of a continuous-time Markov chain on $m$ transient states, started from the probability row vector $\bm\alpha$, with sub-generator matrix $\mathbf{T}$ and exit-rate vector $\mathbf{t}=-\mathbf{T}\mathbf{1}$. Its Laplace transform is rational:
\[
\hat g(\theta) := \E\lb[e^{-\theta U}\rb] = \bm\alpha(\theta I-\mathbf{T})^{-1}\mathbf{t} = \frac{N(\theta)}{D(\theta)},
\]
with $D(\theta)=\det(\theta I-\mathbf{T})$ monic of degree $m$, $\deg N\le m-1$, and $N(0)=D(0)$ since $\hat g(0)=1$. Without loss of generality, we take the representation to be minimal, so that $N$ and $D$ share no root. We highlight two examples: the exponential law $\Exp(\eta)$ is the one-phase case, $D=\theta+\eta$ and $N=\eta$; and the hyperexponential mixture $\sum_{i=1}^m p_i\Exp(\eta_i)$ has diagonal $\mathbf{T}$, $D=\prod_i(\theta+\eta_i)$ and $N=\sum_i p_i\eta_i\prod_{j\ne i}(\theta+\eta_j)$ (we use a mixture with $m=2$ in Section~\ref{sec:common-wh}). We present our key computational result:

\begin{theorem}[Phase-type default probabilities]\label{thm}
Let $X_0^i=-\delta$ with $\delta\ge 0$. Suppose $X^i$ is a spectrally positive L\'{e}vy process with volatility $\sigma\ge0$, drift $\mu\in\R$ ($\mu<0$ if $\sigma=0$), and a compound Poisson component with intensity $\lambda_J>0$ and i.i.d.\ jumps $U\sim\mathrm{PH}(\bm\alpha,\mathbf{T})$ on $m$ phases, with minimal jump transform $\hat g=N/D$ as above. For $q>0$, define
\begin{equation}\label{eq:poly}
P_q(\beta) \;=\;
\lb(\tfrac{1}{2}\sigma^2\beta^2 - \mu\beta - (\lambda_J + q)\rb)D(\beta) + \lambda_J\,N(\beta),
\end{equation}
and set
\[
d \;:=\; \begin{cases} m+2, & \sigma>0,\\ m+1, & \sigma=0.\end{cases}
\]
Then $\deg P_q = d$, $P_q$ has exactly one root in the open right half-plane $\{\real\beta>0\}$, which is $\beta_1(q)$, and the remaining $d-1$ roots lie in the open left half-plane. For $q>0$ such that the small roots $\beta_2(q),\dots,\beta_d(q)$ are simple, we have that
\begin{equation}\label{eq:thm}
\int_0^\infty e^{-qT}\,\bQ(\tau_i\le T)\,dT
\;=\;
\sum_{j=2}^{d}
\frac{
A_j(q)\bigl(\beta_1(q)-\beta_j(q)\bigr)
}{
\beta_1(q)\beta_j(q)\bigl(1-\beta_j(q)\bigr)
}
e^{\beta_j(q)\delta},
\qquad
A_j(q) \;=\; \frac{D(\beta_j(q))}{P_q'(\beta_j(q))}.
\end{equation}
\end{theorem}

Before the proof, we point out that for $\sigma=0$ the requirement that $X^i$ is spectrally positive forces $\mu<0$. This gives exactly the Cram\'er--Lundberg risk process of the introduction, and it is priced by \eqref{eq:thm} with $d=m+1$. Subsection~\ref{sec:sigma-zero} compares the Cram\'er--Lundberg case with the $\mu\ge0$ subordinator case.

\begin{proof}
The Laplace exponent of $-X^i$ is
\[
\psi_i(\theta) \;=\; -\mu\,\theta + \tfrac{1}{2}\sigma^2\theta^2 - \lambda_J\lb(1 - \hat g(\theta)\rb),
\]
defined for $\real\theta$ to the right of the poles of $\hat g$. Clearing the denominator, the equation $\psi_i(\beta)=q$ is equivalent to $P_q(\beta)=0$; minimality ensures no root of $P_q$ is cancelled. The leading coefficient of $P_q$ is $\tfrac12\sigma^2$ when $\sigma>0$ and $-\mu>0$ when $\sigma=0$, since $\mu<0$ in that case, so $\deg P_q=d$ as claimed. \cite[Proposition 5.4]{cohen_theory_2012}, stated there for spectrally negative L\'evy processes with rational jump transforms (the convention our $\psi_i$ already follows, being the Laplace exponent of $-X^i$), gives that $\beta_1(q)$ is simple and the roots are separated as stated.

Now, take $q>0$ and assume the left-half-plane roots are simple. The defining transform for $W_i^{(q)}$ gives
\[
\int_0^\infty e^{-\beta x}W_i^{(q)}(x)\,dx
=
\frac{1}{\psi_i(\beta)-q}
=
\frac{D(\beta)}{P_q(\beta)}.
\]
Since $\deg D = m < d$, the rational function $D/P_q$ is proper and the partial-fraction decomposition is
\[
\frac{D(\beta)}{P_q(\beta)}
=
\sum_{j=1}^{d}
\frac{A_j(q)}{\beta-\beta_j(q)},
\qquad
A_j(q)=\frac{D(\beta_j(q))}{P_q'(\beta_j(q))}.
\]
Inverting term-by-term yields
\begin{equation}\label{eq:Wpf}
W_i^{(q)}(x)
=
\sum_{j=1}^{d}
A_j(q)e^{\beta_j(q)x},
\qquad x\ge0.
\end{equation}
Evaluating the partial-fraction expansion at $\beta=0$, and using $N(0)=D(0)$ so that $P_q(0)=-qD(0)$, gives
\[
-\sum_{j=1}^{d}\frac{A_j(q)}{\beta_j(q)}
=
\frac{D(0)}{P_q(0)}
=
-\frac{1}{q},
\]
and hence
\[
\sum_{j=1}^{d}\frac{A_j(q)}{\beta_j(q)}
=
\frac{1}{q}.
\]
Therefore
\begin{equation}\label{eq:Zpf}
Z_i^{(q)}(x)
=
1+q\int_0^x W_i^{(q)}(u)\,du
=
q\sum_{j=1}^{d}
\frac{A_j(q)}{\beta_j(q)}
e^{\beta_j(q)x},
\qquad x\ge0.
\end{equation}

Using $X_0^i=-\delta$ with $\delta\ge0$ in Lemma~\ref{lem:laplace} gives
\[
\int_0^\infty e^{-qT}\,\bQ(\tau_i\le T)\,dT
=
\int_0^\infty
\frac{e^{-y}}{q}
\lb[
Z_i^{(q)}(y+\delta)
-
\frac{q}{\beta_1(q)}W_i^{(q)}(y+\delta)
\rb]\,dy.
\]
Putting \eqref{eq:Wpf} and \eqref{eq:Zpf} together yields
\[
Z_i^{(q)}(y+\delta)
-
\frac{q}{\beta_1(q)}W_i^{(q)}(y+\delta)
=
q\sum_{j=1}^{d}
A_j(q)
e^{\beta_j(q)(y+\delta)}
\frac{\beta_1(q)-\beta_j(q)}
{\beta_1(q)\beta_j(q)}.
\]
The term $j=1$ vanishes. For $j\ge2$ we have $\real\beta_j(q)<0$, so
\[
\int_0^\infty e^{-y}e^{\beta_j(q)y}\,dy
=
\frac{1}{1-\beta_j(q)},
\]
which proves \eqref{eq:thm}.
\end{proof}

\begin{remark}
For $\Exp(\eta)$ jumps ($m=1$), the polynomial \eqref{eq:poly} is a cubic with one big and two small roots, available in closed form by Cardano's formula. At $\sigma=0$ this degenerates to the quadratic $-\mu\beta^2-(\mu\eta+\lambda_J+q)\beta-q\eta$ with one big and one small root. For general $m$ the roots are computed numerically. When $(\mu,\sigma,\lambda_J,\bm\alpha,\mathbf{T})$ are fixed (as in the panel calibration of Subsection~\ref{subsec:panel}) the root-finding is precomputed once, and a change in the state $\delta$ only changes the factors $e^{\beta_j(q)\delta}$.
\end{remark}

\begin{remark}\label{rmk:distinct}
The simple-root hypothesis is only a convenience for writing \eqref{eq:thm} as a sum over simple poles. The transform itself is well-defined at root coalescences, but in such cases the residues must be computed. The right-hand side of \eqref{eq:thm} is meromorphic in $q\in\mathbb{C}$; Theorem~\ref{thm} therefore provides the meromorphic continuation of the Laplace transform ${\int_0^\infty e^{-qT}\,\bQ(\tau_i\le T)\,dT}$ from its natural domain $\mathbb{C}_+$ to all of $\mathbb{C}$ minus the small-root coalescence locus and the isolated points where some $\beta_j(q)\in\{0,1\}$.
For the numerical Laplace inversion, we choose a Talbot contour (details in Appendix~\ref{sec:Talbot}) with nodes that avoid these bad points.
\end{remark}

\begin{remark}
The proof of Theorem~\ref{thm} uses only the rationality of the jump transform, so it carries through verbatim for any jump law with rational Laplace transform (a slightly larger class than phase-type). Meromorphic L\'evy processes lead to infinite, rather than finite, partial-fraction expansions. See \cite[Sections 5.4 and 5.5]{cohen_theory_2012} for the rational and meromorphic classes, along with a wealth of further references.
\end{remark}

\subsection{The Cram\'er--Lundberg and Subordinator Cases}\label{sec:sigma-zero}

We consider here the two cases of \eqref{eq:driver} with $\sigma=0$: the Cram\'er--Lundberg risk process for $\mu<0$ and the subordinator case for $\mu\ge 0$.
For the Cram\'er--Lundberg case with $\Exp(\eta)$ jumps, the transform \eqref{eq:thm} gives
\[
\int_0^\infty e^{-qT}\,\bQ(\tau_i\le T)\,dT
=
\frac{(\eta+\beta_2(q))\,e^{\beta_2(q)\delta}}{\eta\,q\,(1-\beta_2(q))}.
\]
The subordinator case is excluded from Theorem~\ref{thm} because its paths are non-decreasing. Since it is the model of the previous literature and one of our benchmarks below, we record here precisely what separates it from its $\mu<0$ neighbour. Write $S(T)=\bQ(\tau_i>T)$ for the model survival function and
\[
H(T) \;=\; -\frac{1}{T}\log S(T)
\]
for the average forward hazard to maturity $T$. The following elementary lemma holds for \emph{any} L\'{e}vy subordinator driver.

\begin{lemma}\label{lem:plateau}
Let $X^i=-\delta+\wt X^i$ with $\delta\ge0$ and $\wt X^i$ a L\'evy subordinator started from zero and with Laplace exponent $\varphi$. Then
\[
e^{-\varphi(1)T} \;\le\; S(T) \;\le\; \min\lb(1,\; e^{\delta}e^{-\varphi(1)T}\rb),
\qquad\text{equivalently}\qquad
\varphi(1)-\frac{\delta}{T} \;\le\; H(T) \;\le\; \varphi(1),
\]
so $H(T)\to\varphi(1)$ as $T\to\infty$. 
\end{lemma}

Note that the driver \eqref{eq:driver} with $\sigma=0$ and $\mu\ge0$ is the case $\varphi(\theta)=\mu\theta+\lambda_J(1-\hat g(\theta))$ which satisfies $\varphi(1)=\mu+\lambda_J/(\eta+1)$.

\begin{proof}
Since $X$ is non-decreasing, $\Lambda_T=(-\delta+\wt X_T)_+$ and \eqref{eq:pd} gives $S(T)=\E[e^{-(\wt X_T-\delta)_+}]$. From $\wt X_T-\delta\le(\wt X_T-\delta)_+\le \wt X_T$ and $\E[e^{-\theta \wt X_T}]=e^{-T\varphi(\theta)}$ we obtain $e^{-T\varphi(1)}\le S(T)\le e^{\delta}e^{-T\varphi(1)}$, and $S\le1$ always. Take logarithms and divide by $-T$.
\end{proof}

The next lemma demonstrates the consequence for the shape of the curve, and contrasts it with the $\mu<0$ case under the classical Cram\'er--Lundberg net-profit condition $\E[X_1]=\mu+\lambda_J\E[U]<0$.

\begin{lemma}\label{lem:inversion}
Let $\sigma=0$ and let $X$ be given by \eqref{eq:driver} with $X_0=-\delta$, $\delta\ge0$. If $\mu\ge0$, then
\[
H(T_1)-H(T_2)\;\le\;\frac{\delta}{T_2},\qquad 0<T_1<T_2.
\]
If $\mu<0$ and $\E[X_1]<0$, then $\Lambda_\infty<\infty$ almost surely, $S(\infty)=\E[e^{-\Lambda_\infty}]>0$, and $H(T)\to0$ as $T\to\infty$.
\end{lemma}

\begin{proof}
The first claim subtracts the two bounds of Lemma~\ref{lem:plateau}: $H(T_1)\le\varphi(1)$ and $H(T_2)\ge\varphi(1)-\delta/T_2$. For the second, the net-profit condition gives $X_t\to-\infty$ almost surely (see \cite[Theorem 7.2]{kyprianou_fluctuations_2014}), so $\sup_{t\ge0}X_t$ is finite almost surely and $S(T)\downarrow S(\infty)=\E[e^{-\Lambda_\infty}]>0$ by monotone convergence; then $H(T)=-\log S(T)/T\to0$.
\end{proof}

We state these lemmas formally because the separation they express underlies the empirical results of Sections~\ref{sec:calibration} and \ref{sec:joint-calibration}, and because it gives a falsifiable test. Note that Lemma~\ref{lem:plateau}, and with it the first bound of Lemma~\ref{lem:inversion}, holds for every L\'evy subordinator driver, its proof using only the monotone paths and the Laplace exponent. The mechanism is worth saying in words. Default in \eqref{eq:taudef} requires the distress process to set a new record high; for a subordinator every instant sets or matches the record, so the accumulated hazard can only climb, whereas under the net-profit condition the supremum of the Cram\'{e}r--Lundberg process is almost surely finite and records become rarer as time passes.

In market terms, an inverted CDS spread curve, $S_{\mathrm{par}}(1\mathrm{y})>S_{\mathrm{par}}(5\mathrm{y})$, is the signature of an acutely distressed firm which is expected either to default soon or to survive (see the clinical studies of distressed names in \cite{cherubini_accounting_2008}). Par spreads are increasing transforms of average hazards up to discount weighting, so an inversion requires $H(1y)-H(5y)$ to be large and positive. The first bound of Lemma~\ref{lem:inversion} says that a subordinator driver cannot produce this beyond the buffer term $\delta/T_2$, which vanishes as the fitted distance to distress shrinks; the second says that the Cram\'er--Lundberg driver has finite total hazard $-\log S(\infty)$ and an average hazard decaying to zero, so that it is possible to fit decreasing segments of arbitrary relative size. Figure~\ref{fig:hazard-slope} illustrates the separation, which is verified empirically through market data in Subsection~\ref{sec:inversion}.

\begin{remark}\label{rmk:cirhaz}
The affine benchmarks in Section~\ref{sec:calibration} (CIR and JCIR) can also fit a decreasing average hazard: an intensity started above its long-run level mean-reverts downward. With the statics frozen, that channel is fixed, as demonstrated by the out-of-sample inversion test in Subsection~\ref{sec:inversion}.
\end{remark}

\begin{figure}[h!]
\centering
\includegraphics[width=0.8\textwidth]{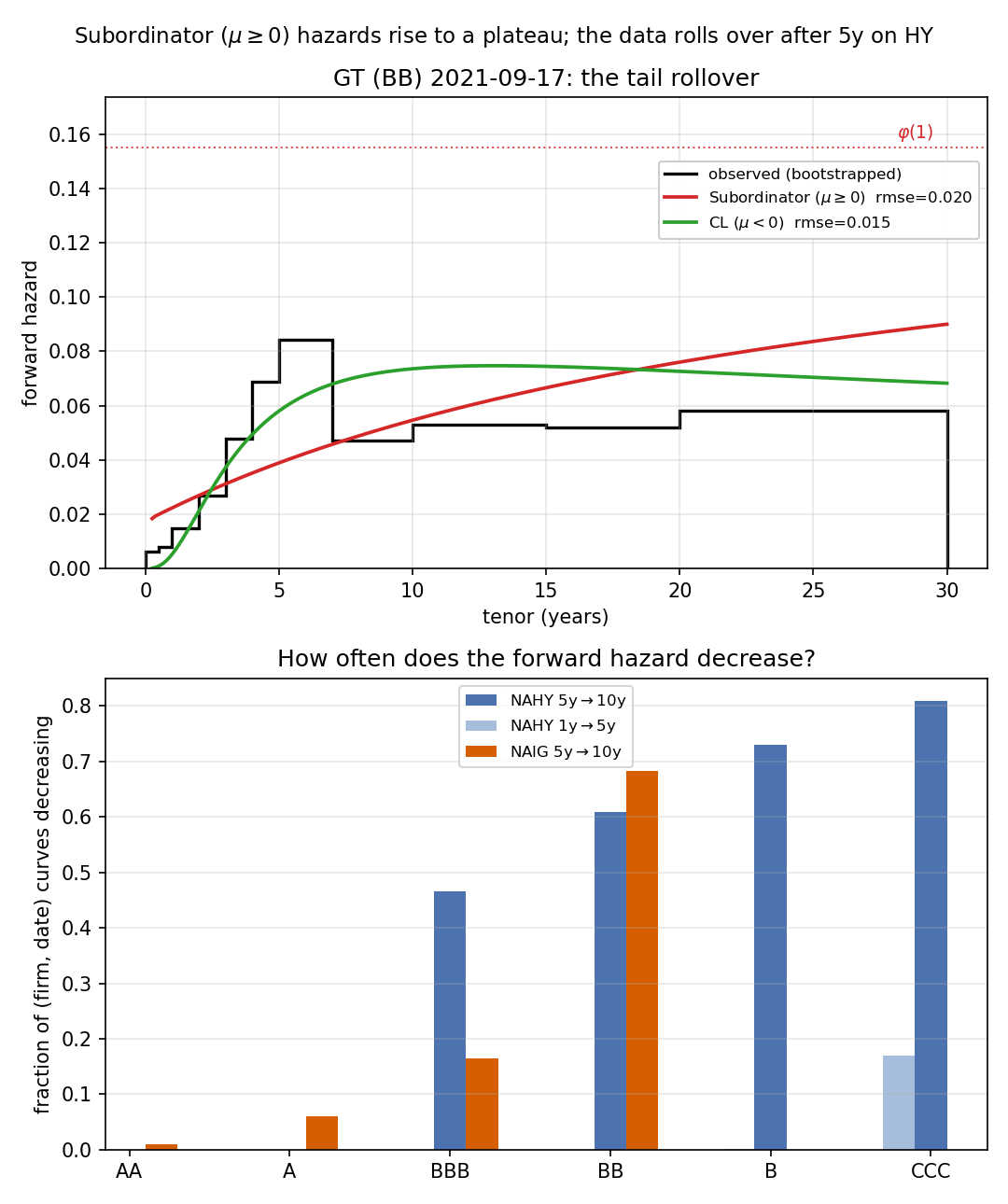}
\caption{The structural separation of Lemmas~\ref{lem:plateau} and~\ref{lem:inversion}. Top: a BB name whose observed forward hazards hump at 5--7y and roll over; the fitted subordinator climbs through the data towards its plateau $\varphi(1)$, while the fitted Cram\'er--Lundberg rolls over with the curve. Bottom: the frequency of decreasing forward-hazard segments in the bootstrapped curve panels, by index and segment, computed over the two years of the spread-inversion test of Subsection~\ref{sec:inversion} (2021-05-27 to 2023-06-30), on the high-yield Series 36 cohort and the investment-grade full basket (definitions in Section~\ref{sec:data} below).}
\label{fig:hazard-slope}
\end{figure}

\section{Common Jumps and a Wiener--Hopf Monte Carlo Scheme}\label{sec:common-wh}

The single-name theory of Section~\ref{sec:onesided} extends to a joint model of many firms at almost no additional cost, and we develop such a model here. We introduce a single common jump factor that preserves the phase-type marginals, and an exact simulation scheme whose clock is well-suited for the model. 

\subsection{Common Jump Construction}\label{sec:joint-construction}

Consider $n$ firms, each carrying one of the three subcases of \eqref{eq:driver} discussed in Section~\ref{sec:onesided} with its own idiosyncratic driver and initial distance-to-distress $\delta_i$. Let $Y$ be a single compound Poisson subordinator, shared by every firm and independent of the idiosyncratic drivers, with arrival intensity $\lambda_c>0$ and $\Exp(1/\gamma)$ jump sizes, so that $\gamma>0$ is the mean size of a common jump. The joint model adds the common factor to each firm's driver before the supremum is taken: firm $i$'s cumulative hazard and default time are
\begin{equation}\label{eq:joint-common}
\Lambda^i_t \;=\; \sup_{s\le t}\,\bigl(\wt X^i_s + Y_s\bigr),
\qquad
\tau_i \;=\; \inf\lb\{t\ge0:\ \Lambda^i_t \ge \delta_i + \xi_i\rb\}.
\end{equation}
Each such process is non-decreasing and adapted, so nothing in Section~\ref{sec:onesided} needs modification: a firm defaults when its distress (its own shocks plus the accumulated systemic ones) sets a sufficiently high record. For a single firm with $Y\equiv0$, \eqref{eq:joint-common} is exactly the construction of \eqref{eq:taudef}--\eqref{eq:Lambdadef} with $X_0^i=-\delta_i$, since $\xi_i>0$ almost surely. We choose to put $\delta_i$ with the barrier so that the supremum can be stated without truncation, and so that $\Lambda^i$ is a L\'{e}vy process started from zero, which is more natural for the scheme in the next subsection. 

Two features of \eqref{eq:joint-common} are worth distinguishing. The first concerns simultaneous defaults. At each arrival of $Y$, every firm's distress jumps by the same $\Exp(1/\gamma)$ amount, so any number of firms may default at that instant with positive probability. This is a singular component of Marshall--Olkin type, of the kind \cite{bo_counterparty_2015} and \cite{sun_marshallolkin_2017} obtain within the subordinator class. The second feature is that the construction keeps the single-name marginals within the class of Theorem~\ref{thm}. Seen from a single firm, the driver $\wt X^i+Y$ is a L\'evy process nearly of the form \eqref{eq:driver}: superimposing the firm's own jump stream, $\Exp(\eta)$ at rate $\lambda_J$, with the common one, $\Exp(1/\gamma)$ at rate $\lambda_c$, gives a compound Poisson process of total intensity $\lambda_J+\lambda_c$ whose jump law is the two-component mixture
\begin{equation}\label{eq:common-marginal}
\frac{\lambda_J}{\lambda_J+\lambda_c}\,\Exp(\eta)
\;+\;
\frac{\lambda_c}{\lambda_J+\lambda_c}\,\Exp(1/\gamma),
\end{equation}
which is phase-type with $m=2$. Theorem~\ref{thm} therefore prices the single-name marginals of the joint model exactly: at $d=4$ for LBM+Exp and $d=3$ for Cram\'er--Lundberg. The subordinator case is priced in closed form using its Laplace transform. The mixture's new ingredients are fixed by the dependence parameters rather than fitted separately, the second component having mean $\gamma$ and weight $\lambda_c/(\lambda_J+\lambda_c)$.

\subsection{Wiener--Hopf Monte Carlo Scheme}

Pricing a tranche on the joint law \eqref{eq:joint-common} requires, for every firm, the running supremum and the current position of its driver at a set of deterministic dates. A grid scheme is poorly suited to this: a discretization of mesh $\Delta t$ sees only the grid-point values, missing the excursion maxima between them, and so underestimates the supremum with a systematic bias of order $\sqrt{\Delta t}$ (see \cite{asmussen_discretization_1995}). The scheme of this subsection, which follows the Wiener--Hopf Monte Carlo approach of \cite{kuznetsov_wienerhopf_2011}, removes the time grid altogether. It makes use of three observations: the law of the supremum-position pair at an independent exponential time is explicit for our drivers; repeating exponential steps reaches deterministic dates; and the common compound Poisson factor of \eqref{eq:joint-common} arrives precisely at such an exponential time.

The first observation is known as the \emph{Wiener--Hopf factorization} at an independent exponential time. Let $\wt X$ be a driver from the class of Theorem~\ref{thm} started at zero, let $e_q\sim\Exp(q)$ be independent of it, and write $\wt\Lambda_{e_q}=\sup_{s\le e_q}\wt X_s$. The factorization (see \cite[Chapter 6]{kyprianou_fluctuations_2014}) shows that the supremum and the drawdown at the exponential time are independent, so that
\begin{equation}\label{eq:wh-pair}
\bigl(\wt\Lambda_{e_q},\ \wt X_{e_q}\bigr) \;\overset{d}{=}\; \bigl(M,\ M-D\bigr),
\qquad M\ \text{and}\ D\ \text{independent},
\end{equation}
where $M$ is distributed as $\wt\Lambda_{e_q}$ and $D$ as the drawdown $\wt\Lambda_{e_q}-\wt X_{e_q}$. For spectrally positive drivers both laws are explicit in the roots of the polynomial $P_q$ of Theorem~\ref{thm}. The drawdown is distributed as the supremum at $e_q$ of the dual process $-\wt X$, which is spectrally negative and therefore creeps upward, making that supremum exponential:
\[
D \;\sim\; \Exp(\beta_1(q))
\]
for any phase-type jump structure (\cite{pistorius_maxima_2006}; see also \cite[Section 8.1]{kyprianou_fluctuations_2014}). The supremum factor has the tail
\[
\bQ(M>x) \;=\; q\sum_{j=2}^{d} \frac{A_j(q)\,(\beta_1(q)-\beta_j(q))}{\beta_1(q)\,\beta_j(q)}\, e^{\beta_j(q)x},
\qquad x\ge0,
\]
by Lemma~\ref{lem:0init} and the partial fractions \eqref{eq:Wpf}--\eqref{eq:Zpf}.
Sampling the pair \eqref{eq:wh-pair} exactly therefore takes \emph{two independent scalar draws}: one of $M$ from the above, and one of $D$ from an exponential. Both are built from the same root system that the Talbot inversion of Appendix~\ref{sec:Talbot} already computes, so the sampler costs no additional per-model mathematics.

The second observation, due to \cite{kuznetsov_wienerhopf_2011}, is that chaining such steps approximates a deterministic horizon. Let $(M_k, D_k)_{k\ge1}$ be i.i.d.\ copies of the pair in \eqref{eq:wh-pair} and set, from $\wt X_{(0)}=\wt\Lambda_{(0)}=0$,
\begin{equation}\label{eq:wh-chain}
\wt X_{(k)} \;=\; \wt X_{(k-1)}+M_k-D_k,
\qquad
\wt\Lambda_{(k)} \;=\; \max\lb(\wt\Lambda_{(k-1)},\; \wt X_{(k-1)}+M_k\rb).
\end{equation}
By the stationarity and independence of L\'evy increments, the supremum over the $k$-th exponential segment is the position at its start plus an independent copy of $\wt\Lambda_{e_q}$, so $(\wt\Lambda_{(n)},\wt X_{(n)})$ has exactly the law of the pair at the $\mathrm{Erlang}(n,q)$ time $e_q^{(1)}+\dots+e_q^{(n)}$. No step of \eqref{eq:wh-chain} discretizes time; the only systematic error at a deterministic maturity $T$ is that the scheme evaluates at an Erlang time concentrated near $T$ rather than at $T$ itself, with relative width of order $1/\sqrt{qT}$. The base rate $q$ of the clock therefore exists to resolve deterministic maturities, and refining $q$ trades more steps for a tighter random grid.

The third observation is specific to our joint model: the compound Poisson common factor fires at exponential times, and between two consecutive arrivals each firm moves by its idiosyncratic increments alone. Adding the pricing clock rate and the arrival clock rate yields segments of i.i.d.\ $\Exp(q+\lambda_c)$ length, each ending, independently of the past, in a common arrival with probability $\lambda_c/(q+\lambda_c)$ and in a pricing step otherwise. On each segment we apply the recursion \eqref{eq:wh-chain} with $q$ replaced by $q+\lambda_c$, independently across firms: conditional on the path of $Y$ the firms are independent, so one shared segment grid serves the whole portfolio, with two draws per firm per segment. If the segment ends in a common arrival, a single $\Exp(1/\gamma)$ variable is drawn and added to every firm's position at once, and each running supremum is refreshed where the jump sets a new record. The common factor's event times are thus resolved exactly at any $q$, while the pricing steps alone still form a Poisson process of rate $q$, so the horizon remains Erlang-concentrated as before. Any driver with a common compound Poisson shock and a marginal covered by Theorem~\ref{thm} can be simulated this way. A step-rate sweep reported with the numerical controls in Section~\ref{sec:joint-calibration} confirms that $q=16$ per year suffices at the portfolio level.

The speed of the above scheme is best demonstrated against a fine-grid Euler simulation benchmark of the same joint law. Table~\ref{tab:r10-bench} prices a full index tranche portfolio ($95$ names, at parameters fitted in Section~\ref{sec:joint-calibration}; nothing in the comparison depends on those details) by both routes. The Wiener--Hopf scheme's accuracy is governed by the path count alone: it reaches $0.011$ in under a second and $0.002$ in fourteen. The Euler scheme must refine the grid and the path count together, and the accuracy it attains after roughly five hundred seconds, about $0.003$, is beaten by the Wiener--Hopf scheme in just fourteen seconds.

\begin{table}[t]
\centering\small
\begin{tabular}{lllrr}
\toprule
Scheme & Paths $K$ & Step & Wall (s) & $\max|\mathrm{dev}|$ \\
\midrule
WHMC (event-driven) & 2,000 & $q=16$ & 0.76 & 0.0112 \\
WHMC (event-driven) & 2,000 & $q=48$ & 2.18 & 0.0212 \\
WHMC (event-driven) & 8,000 & $q=16$ & 2.90 & 0.0067 \\
WHMC (event-driven) & 8,000 & $q=48$ & 8.69 & 0.0068 \\
WHMC (event-driven) & 32,000 & $q=16$ & 13.79 & 0.0022 \\
WHMC (event-driven) & 32,000 & $q=48$ & 40.34 & 0.0020 \\
\addlinespace
Euler (fine grid) & 2,000 & $\Delta t = 1/252$ & 7.55 & 0.0058 \\
Euler (fine grid) & 2,000 & $\Delta t = 1/1008$ & 29.79 & 0.0129 \\
Euler (fine grid) & 2,000 & $\Delta t = 1/4032$ & 119.18 & 0.0063 \\
Euler (fine grid) & 8,000 & $\Delta t = 1/252$ & 32.75 & 0.0082 \\
Euler (fine grid) & 8,000 & $\Delta t = 1/1008$ & 136.16 & 0.0048 \\
Euler (fine grid) & 8,000 & $\Delta t = 1/4032$ & 527.32 & 0.0032 \\
Euler (fine grid) & 32,000 & $\Delta t = 1/252$ & 128.92 & 0.0074 \\
Euler (fine grid) & 32,000 & $\Delta t = 1/1008$ & 514.03 & 0.0027 \\
Euler (fine grid) & 32,000 & $\Delta t = 1/4032$ & 2048.48 & 0.0052 \\
\bottomrule
\end{tabular}
\caption{Joint tranche pricing, Wiener--Hopf Monte Carlo vs fine-grid Euler on the same index tranche portfolio, both schemes simulating the joint law \eqref{eq:joint-common}; mean wall time and mean maximum absolute tranche-upfront deviation from the converged reference, over common-random-number seeds. The Wiener--Hopf scheme performs exact sampling at exponential steps, so its accuracy is set by the path count alone; at matched accuracy it is one to two orders of magnitude faster.}
\label{tab:r10-bench}
\end{table}

\section{Data}\label{sec:data}

We work with the constituents and tranches of two credit indices: CDX-NAHY (about 100 North American high-yield names) and CDX-NAIG (about 125 investment-grade names). All data were obtained from S\&P Global through Wharton Research Data Services (WRDS).

\paragraph{Single-name panels.}
Daily single-name quotes come from the Markit composites (note that Markit is now part of S\&P Global at the time of writing). Table~\ref{tab:data-singlename} summarizes the two panels, which run from 2021-05-27 to 2025-02-13. 
The span is scored out of sample in rolling quarterly folds (details in Subsection~\ref{sec:evaldesign}); the first year, the $248$ trading days ending 2022-05-27, additionally serves as an identification window for in-sample comparison.

\begin{table}[p]
\centering
\small
\begin{tabular}{lll}
\toprule
 & NAHY & NAIG \\
\midrule
Firms per date (median) & 98 & 125 \\
Distinct firms over window & 137 & 152 \\
Firm-dates & 93,458 & 119,364 \\
5y spread p5/p50/p95 (bp) & 87 / 285 / 1464 & 25 / 56 / 156 \\
Recovery & 0.30 & 0.40 \\
Running coupon & 500\,bp & 100\,bp \\
\bottomrule
\end{tabular}
\caption{Single-name CDS panels (Markit), 2021-05-27 to 2025-02-13 for both indices, the end date being where the treasury series stops. At each date they carry the constituents of the then-current index, so names enter and leave at roll-dates and the firm count is the number appearing at a given date.}
\label{tab:data-singlename}
\end{table}

\begin{table}[p]
\centering\small
\begin{tabular}{lrrrrrrr}
\toprule
Sector & AA & A & BBB & BB & B & CCC & All \\
\midrule
\multicolumn{8}{l}{\textbf{NAHY}} \\
Basic Materials & --- & 1 & 2 & 5 & 2 & 1 & 11 \\
Consumer Goods & --- & --- & 4 & 8 & 4 & 1 & 17 \\
Consumer Services & --- & --- & 3 & 11 & 11 & 10 & 35 \\
Energy & --- & --- & 2 & 7 & --- & 2 & 11 \\
Financials & --- & --- & 4 & 6 & 6 & 2 & 18 \\
Healthcare & --- & --- & 1 & 4 & 2 & 2 & 9 \\
Industrials & --- & --- & 1 & 9 & 4 & 1 & 15 \\
Technology & --- & --- & 1 & 2 & 3 & 3 & 9 \\
Telecom. Services & --- & --- & 1 & --- & 2 & --- & 3 \\
Utilities & --- & --- & 1 & 5 & 1 & 2 & 9 \\
All & --- & 1 & 20 & 57 & 35 & 24 & 137 \\
\addlinespace
\multicolumn{8}{l}{\textbf{NAIG}} \\
Basic Materials & 1 & 5 & 4 & --- & 1 & --- & 11 \\
Consumer Goods & 1 & 5 & 9 & 4 & --- & --- & 19 \\
Consumer Services & 4 & 13 & 6 & 1 & 1 & --- & 25 \\
Energy & --- & 2 & 13 & 2 & --- & --- & 17 \\
Financials & 12 & 2 & 6 & 3 & --- & --- & 23 \\
Healthcare & 1 & 7 & 2 & --- & --- & --- & 10 \\
Industrials & 9 & 8 & 4 & 1 & --- & --- & 22 \\
Technology & 2 & 5 & 3 & 1 & --- & --- & 11 \\
Telecom. Services & --- & --- & 3 & --- & --- & --- & 3 \\
Utilities & 1 & 6 & 3 & 1 & --- & --- & 11 \\
All & 31 & 53 & 53 & 13 & 2 & --- & 152 \\
\bottomrule
\end{tabular}
\caption{Distribution of constituents across sectors and rating buckets. Distinct firms over the window, each firm bucketed at its most frequent Markit rating across its quoted dates. Both populations are the full baskets of Table~\ref{tab:data-singlename}.}
\label{tab:data-sector-rating}
\end{table}

Table~\ref{tab:data-sector-rating} reports the distribution of the constituents across sectors and rating buckets. The high-yield basket concentrates in BB and B names with a substantial CCC tail in consumer services, while the investment-grade basket sits in A and BBB with its AA weight concentrated in financials.

Table~\ref{tab:data-singlename} reports both indices as \emph{full baskets}: at each date the panel carries the constituents of the on-the-run index, so that names enter and leave at the rolls. A basket is faithful to what a desk actually holds, but the population it scores changes through time. The alternative is a \emph{cohort}, a constituent list fixed once and followed thereafter, which is survivorship-prone, since the names dropping out are disproportionately the ones that defaulted. Neither convention is the right one on its own, so on the high-yield index we pulled both datasets, and Table~\ref{tab:data-cohort} sets them side by side. The Series 36 cohort is quoted on $100$ names in 2021 and $68$ by 2025, admitting no replacements, whereas the basket stays near $100$ throughout by taking in $31$ new names as $33$ leave. On the investment-grade index only the basket was pulled, and its turnover is milder, $24$ names in and $24$ out against a near-constant $125$. We report both high-yield universes alongside the single investment-grade basket as separate \emph{boards}, our term throughout for a constituent universe.

\begin{table}[t]
\centering\small
\begin{tabular}{lll}
\toprule
 & Cohort & Full basket \\
\midrule
Firms per date (median) & 78 & 98 \\
Firms quoted, 2021$\to$2025 & 100/97/82/74/68 & 103/115/109/108/100 \\
Distinct names over window & 100 & 137 \\
Names entering & 0 & 31 \\
Names leaving & 33 & 33 \\
Firm-dates & 76,643 & 93,458 \\
\bottomrule
\end{tabular}
\caption{The two NAHY datasets. The cohort is the constituent list of CDX-NAHY Series 36, on the run at the start of the window and held fixed thereafter, so no name ever enters and the decay is corporate-event attrition. The full basket takes the constituents of the then-current series at each date. Cohort results compare a fixed population and are survivorship-prone; full-basket results are faithful to the traded index. NAIG data were pulled only as a basket.}
\label{tab:data-cohort}
\end{table}

\paragraph{Index tranches.}
Tranche quotes come from S\&P Global, filtered per index to the 5Y term and, at each date, to the on-the-run series; the series advances annually, with roll-dates near October 1. The standardized attachment points are $0/15/25/35/100\%$ on the high-yield index and $0/3/7/15/100\%$ on the investment-grade index, and the index recoveries are $0.30$ and $0.40$ respectively; the composite construction and the pricing conventions are documented in \cite{sp_global_market_intelligence_credit_2024}. The dataset includes the upfront for both indices, and we work with this quantity throughout; it is expressed as a fraction of tranche notional and signed so that a positive value is a payment from the protection buyer at inception. Table~\ref{tab:data-tranche} summarizes the two panels. The identities converting a joint default-count law into these upfronts, and the definition of the \emph{tranche error} which serves as our joint calibration target, are collected in Appendix~\ref{sec:tranche-calibration}.

\begin{table}[ht]
\centering
\small
\begin{tabular}{lrrr}
\toprule
Tranche & Quotes & Upfront mean & Upfront range \\
\midrule
\multicolumn{4}{l}{\textbf{NAHY}, series 35/37/39/41/43, 500\,bp coupon} \\
0\%-15\% & 961 & 0.537 & [0.35, 0.70] \\
15\%-25\% & 961 & 0.023 & [-0.09, 0.23] \\
25\%-35\% & 961 & -0.106 & [-0.17, 0.04] \\
35\%-100\% & 961 & -0.181 & [-0.22, -0.14] \\
\addlinespace
\multicolumn{4}{l}{\textbf{NAIG}, series 35/37/39/41/43, 100\,bp coupon} \\
0\%-3\% & 962 & 0.352 & [0.23, 0.54] \\
3\%-7\% & 962 & 0.076 & [0.02, 0.21] \\
7\%-15\% & 962 & -0.005 & [-0.02, 0.04] \\
15\%-100\% & 962 & -0.035 & [-0.04, -0.03] \\
\bottomrule
\end{tabular}
\caption{Index tranche panels; 5Y on-the-run composites with annual series rolls (approximately Oct-1). Upfronts are mid quotes as a fraction of tranche notional; the running coupon is constant within each index and is read from the dataset.}
\label{tab:data-tranche}
\end{table}

\paragraph{Discount curve and the hazard-space target.}
The risk-free discount curve is bootstrapped daily from the FRED constant-maturity Treasury par-yield series. We solve for a zero-coupon curve at standard tenors of $0.5$, $1$, $2$, $3$, $5$, $7$, $10$, $20$, and $30$ years, and use the resulting discount factors to obtain risk-neutral default probabilities from the CDS prices (see, for instance, \cite[Subsection 3.1.4]{brigo_counterparty_2013}).
For each (firm, date)-pair we run an ISDA-style bootstrap with a quarterly premium leg, accrual-on-default, recovery taken from Markit's \texttt{cdsassumedrecovery} field, and the computed discount curve (Appendix~\ref{sec:isda-bootstrap} details the bootstrap methodology). The output is a vector of forward hazard rates at the observed maturities (typically $0.5$, $1$, $2$, $3$, $4$, $5$, $7$, and $10$ years). That vector, and not the quoted spread, is what every model in this paper is fitted against.

\section{Single-Name Calibration}\label{sec:calibration}

\subsection{Model Field and Calibration Protocol}\label{sec:model-field}

We now calibrate the framework of \eqref{eq:taudef}--\eqref{eq:pd} to the single-name panel data described in Section~\ref{sec:data}. The field consists of six models, listed in Table~\ref{tab:model-field}. Two are intensity-based benchmarks, Cox--Ingersoll--Ross (CIR) and CIR with jumps (JCIR); one is the structural Black--Cox model; two are members of our own class from Section~\ref{sec:onesided}, namely LBM+Exp with $\sigma>0$ in \eqref{eq:driver} and Cram\'er--Lundberg with $\sigma=0$ and $\mu<0$; and the sixth is a subordinator benchmark obtained by setting $\sigma=0$ and $\mu\ge 0$ in \eqref{eq:driver}, which lies outside the spectrally positive class by convention.

\begin{table}[h]
\centering
\begin{tabular}{lcll}
\toprule
Model & Params & State & Static \\
\midrule
Black--Cox & 3 & $x_0$ & $\mu,\sigma$ \\
CIR & 4 & $\lambda_0$ & $\kappa,\theta,\sigma$ \\
JCIR & 6 & $\lambda_0$ & $\kappa,\theta,\sigma,\alpha,\gamma$ \\
subordinator & 4 & $\delta$ & $\mu\,(\ge0),\lambda_J,\eta$ \\
Cram\'er--Lundberg (CL) & 4 & $\delta$ & $\mu\,(<0),\lambda_J,\eta$ \\
LBM+Exp & 5 & $\delta$ & $\mu,\sigma,\lambda_J,\eta$ \\
\bottomrule
\end{tabular}
\caption{The six-model field. The state is re-fit daily; the statics are constant per firm. Parameter counts include the state.}
\label{tab:model-field}
\end{table}

Every model is calibrated under one protocol. We split its parameter vector into a per-firm \emph{static} part, held constant across the dates of a fixed window, and a scalar \emph{state} re-fitted at every date. The split is what makes the parameters identifiable: a single CDS curve carries eight tenors and is therefore worth roughly one effective parameter. A static vector of four or five parameters can only be pinned down by pooling dates, while the day-to-day movement of a firm's credit is re-marked by the state. This split also embodies the question we wish to ask of a model: whether one re-marked coordinate is enough to absorb or explain a surprise in pricing.
 
Our calibration loss is the root-mean-squared error between bootstrapped and model-implied forward hazard rates at the observed maturities $\mathcal{T}_{i,t}$,
\[
\mathrm{RMSE}_{i,t}(\theta) = \sqrt{\frac{1}{|\mathcal{T}_{i,t}|}\sum_{T\in\mathcal{T}_{i,t}} \bigl(h^{\mathrm{boot}}_{i,t}(T) - h^\theta_{i,t}(T)\bigr)^2},
\]
where $h^\theta$ denotes the forward hazard implied by $\bQ(\tau\le T)$ under parameters $\theta$. Working in forward-hazard space rather than survival-probability space normalizes the residual scale across maturities; it is also what the bootstrap returns, so no further transformation is required at evaluation time. Alongside RMSE we report the mean absolute error $\mathrm{MAE}_{i,t}(\theta) = |\mathcal{T}_{i,t}|^{-1}\sum_{T\in\mathcal{T}_{i,t}}|h^{\mathrm{boot}}_{i,t}(T) - h^\theta_{i,t}(T)|$ as a robustness check, when appropriate.

Writing $\theta_i$ for firm $i$'s static vector and $s_{i,t}$ for its state on date $t$, the panel objective over a window of dates $\mathcal{D}_i$ is
\begin{equation}\label{eq:panel-loss}
\mathcal{L}_i\bigl(\theta_i, \{s_{i,t}\}\bigr) \;=\; \sum_{t\in\mathcal{D}_i} \mathrm{RMSE}_{i,t}\bigl(\theta_i, s_{i,t}\bigr),
\end{equation}
which we minimize by coordinate descent, alternating a cheap one-dimensional update of each $s_{i,t}$ at fixed $\theta_i$ against an outer minimization over $\theta_i$ at fixed states. Note that a fit of this kind is non-anticipative only on the last date of its window, and so the honest test is given by the out-of-sample designs of Subsection~\ref{subsec:panel}. The window length $\mathcal{D}_i$ is determined by an empirical sweep (see the end of Subsection~\ref{subsec:panel}) and we adopt a trailing three-month window, re-fitted quarterly, in the out-of-sample rolling design that follows. The in-sample comparison of Table~\ref{tab:singlename-field} instead uses the one-year identification window of Section~\ref{sec:data}.

The LBM+Exp, Cram\'er--Lundberg, and subordinator drivers share one parameterization: drift $\mu$, Brownian volatility $\sigma$, and compound Poisson jumps of intensity $\lambda_J$ with $\Exp(\eta)$ sizes, with the per-date state $\delta_t\ge0$ setting the distance to distress. LBM+Exp is the $\sigma>0$ case, at $m=1$ in Theorem~\ref{thm}, and Cram\'er--Lundberg is the $\sigma=0$ case of the same theorem. Both recover the default probability from \eqref{eq:thm} by Talbot inversion (Appendix~\ref{sec:Talbot}), the only difference being the degree of the polynomial to be solved: a cubic with two small roots for LBM+Exp, and a quadratic with one for Cram\'er--Lundberg. The subordinator, being $\sigma=0$ with $\mu\ge0$, needs no inversion at all, since $\Lambda_T=X_T$ and the default probability follows from the law of $X_T$ in closed form. In the two inverted cases the Talbot contour data depends only on the statics, so it is built once per static trial and each per-date evaluation reduces to a sum of complex exponentials $e^{\beta_j\delta_t}$; this is what makes a panel study of this size tractable.

In the Black--Cox model, a firm $i$ defaults at $\tau_i=\inf\lb\{t\ge0:\ x_0^i+\mu t+\sigma B^i_t\le0\rb\}$ where $x_0^i>0$, which admits a survival function in closed form (see \cite{bielecki_credit_2004}). The two reduced-form baselines are standard intensity models, fitting the Cox construction \eqref{eq:taudef} with the cumulative hazard given by the integral $\Lambda^i_t=\int_0^t\lambda^i_s\,ds$ of an intensity process rather than by a running supremum. CIR posits
\[
d\lambda^i_t=\kappa(\theta-\lambda^i_t)\,dt+\sigma\sqrt{\lambda^i_t}\,dW^i_t,
\]
and JCIR adds a compound Poisson term to get
\[
d\lambda^i_t=\kappa(\theta-\lambda^i_t)\,dt+\sigma\sqrt{\lambda^i_t}\,dW^i_t+dJ^i_t,
\]
with arrival rate $\alpha$ and exponentially distributed jump sizes of mean $\gamma$. The appropriate pricing formulas can be found in \cite{brigo_counterparty_2013}, which also contains extensions of the benchmark models with time-varying backbones (\cite{brigo_counterparty_2013} calls these CIR++, JCIR++, and AT1P: Analytically-Tractable First Passage Model). We do not pursue such extensions here. We also relegate the joint modelling of default probabilities and interest rates to future work.

\subsection{Panel Results}\label{subsec:panel}\label{sec:evaldesign}

Table~\ref{tab:singlename-field} reports the model field on both indices, fitted on the one-year identification window using the protocol of Subsection~\ref{sec:model-field}. LBM+Exp records the lowest median error on both, $0.0090$ on high-yield and $0.0017$ on investment grade, followed on high-yield by JCIR at $0.0093$. Cram\'er--Lundberg records $0.0098$ and $0.0017$, within eight basis points of LBM+Exp on high-yield and equal to it on investment grade, with one parameter fewer. The subordinator records $0.0163$ and $0.0048$, which are $1.8$ and $2.8$ times the LBM+Exp error. Since the subordinator differs from Cram\'er--Lundberg only in the sign constraint on $\mu$, the comparison between those two rows isolates the cost of that constraint: $66\%$ of median error on high-yield and $180\%$ on investment grade. The MAE columns preserve the RMSE ordering, except that Black--Cox and CIR trade places on high-yield. These are in-sample figures, and the designs below re-score the same field out of sample.

The Black--Cox figure on investment grade reflects a boundary solution rather than a poor interior fit: on $97\%$ of firm-dates the fitted $x_0$ pins at the limit of its search range, where the implied hazard is negligible at every tenor, so the recorded error is close to the size of the curve being fitted (median root-mean-square $0.0159$), and relaxing the limit leaves the fit unchanged. The shape of the target explains the collapse: investment-grade forward hazards rise monotonically in maturity, whereas the first-passage hazard of a Brownian motion is unimodal.

\begin{table}[t]
\centering\small
\begin{tabular}{lcrrrr}
\toprule
 & & \multicolumn{2}{c}{NAHY} & \multicolumn{2}{c}{NAIG} \\
\cmidrule(lr){3-4}\cmidrule(lr){5-6}
Model & Params & RMSE & MAE & RMSE & MAE \\
\midrule
Black--Cox & 3 & 0.0125 & 0.0104 & 0.0155 & 0.0135 \\
CIR & 4 & 0.0128 & 0.0101 & 0.0030 & 0.0024 \\
JCIR & 6 & 0.0093 & 0.0076 & 0.0022 & 0.0018 \\
subordinator & 4 & 0.0163 & 0.0131 & 0.0048 & 0.0041 \\
CL & 4 & 0.0098 & 0.0079 & \textbf{0.0017} & \textbf{0.0014} \\
LBM+Exp & 5 & \textbf{0.0090} & \textbf{0.0067} & \textbf{0.0017} & \textbf{0.0014} \\
\bottomrule
\end{tabular}
\caption{Single-name in-sample panel calibration, six-model field. Median per-firm RMSE and MAE of fitted vs bootstrapped forward hazards over the one-year panel (statics shared across dates, one state per date). Parameter counts include the state.}
\label{tab:singlename-field}
\end{table}

\paragraph{Rolling-origin evaluation.}
The panel fit is non-anticipative only on its last date. We therefore re-score the field under a design that is non-anticipative at every date. Each of the three boards of Section~\ref{sec:data} (the high-yield Series 36 cohort, the high-yield full basket, and the investment-grade full basket) is scored separately. Statics are re-fitted each quarter on the trailing three months, warm-started from the preceding fold, and scored on the following quarter, giving fourteen folds per board from mid-2021 to early 2025. Within each fold we take the median error across firms, and Table~\ref{tab:rolling-ladder} reports the average of those fourteen numbers. Figure~\ref{fig:fold-series} shows the per-fold medians of each model in sequence. The six series trace one common movement, high-yield error roughly doubling into late 2022 and halving again across 2024, with the spectrally positive pair in a tight band at the bottom throughout. The dotted verticals mark five market events, from the invasion of Ukraine (the onset of the 2022 climb) to the yen-carry unwind. The full-basket panel traces the same rise and decline as the cohort panel, so the easing across 2024 is a movement of the market rather than an artifact of attrition in the cohort. The high-yield ordering of Table~\ref{tab:singlename-field} is unchanged; on investment grade, JCIR edges ahead of Cram\'er--Lundberg and the subordinator falls to last. LBM+Exp records the lowest error on all three boards, JCIR is within two basis points of it on high-yield, Cram\'er--Lundberg within eight, and the subordinator records the highest error on each. Model by model, the full-basket board runs between three and six per cent above its cohort counterpart.

\begin{table}[t]
\centering
\small
\begin{tabular}{lrrrrrr}
\toprule
 & \multicolumn{2}{c}{NAHY cohort} & \multicolumn{2}{c}{NAHY full} & \multicolumn{2}{c}{NAIG} \\
\cmidrule(lr){2-3}\cmidrule(lr){4-5}\cmidrule(lr){6-7}
Model & RMSE & MAE & RMSE & MAE & RMSE & MAE \\
\midrule
Black--Cox & 0.0134 & 0.0111 & 0.0141 & 0.0117 & 0.0032 & 0.0027 \\
CIR & 0.0140 & 0.0111 & 0.0145 & 0.0116 & 0.0026 & 0.0021 \\
JCIR & 0.0100 & 0.0076 & 0.0105 & 0.0081 & 0.0019 & 0.0015 \\
subordinator & 0.0149 & 0.0116 & 0.0154 & 0.0120 & 0.0052 & 0.0045 \\
CL & 0.0106 & 0.0084 & 0.0112 & 0.0090 & 0.0020 & 0.0016 \\
LBM+Exp & \textbf{0.0098} & \textbf{0.0074} & \textbf{0.0104} & \textbf{0.0079} & \textbf{0.0018} & \textbf{0.0014} \\
\bottomrule
\end{tabular}
\caption{Rolling-origin evaluation. Mean over folds of the median out-of-sample RMSE and MAE, per board. The full-basket board scores every firm quoted in each fold window ($96$--$99$ firms per fold, against the cohort\textquotesingle s $68$--$98$) rather than the fixed survivor cohort.}
\label{tab:rolling-ladder}
\end{table}

\begin{figure}
\centering
\includegraphics[width=\textwidth]{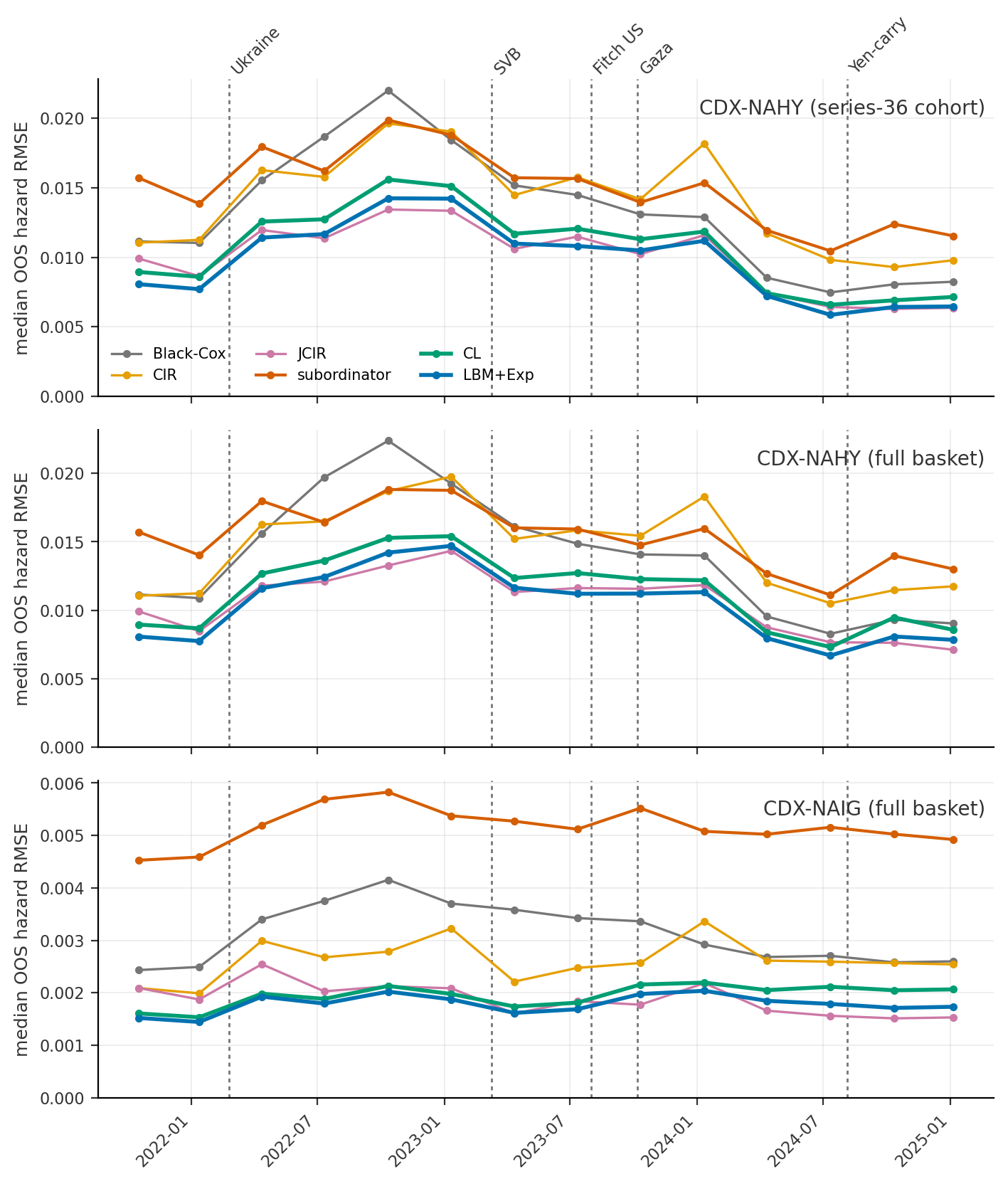}
\caption{Per-fold out-of-sample error across the rolling folds. The per-fold median hazard RMSE of each model, with dotted verticals at five market-event onsets: the invasion of Ukraine, the Silicon Valley Bank implosion, the Fitch US credit downgrade, the onset of the Gaza conflict, and the yen carry-trade unwind. Top panel: the high-yield Series 36 cohort; middle panel: the high-yield full basket; bottom panel: the investment-grade full basket. Note the different vertical scales.}
\label{fig:fold-series}
\end{figure}

\paragraph{Staleness and window length.}
We examine two questions pertaining to our experimental design. The first is how quickly a fit decays once made. Table~\ref{tab:rolling-staleness} divides the quarterly folds into three age buckets. On high-yield the error grows monotonically with age for every model, by between $11\%$ and $26\%$ across a three-month fold, the affine models degrading fastest and Black--Cox slowest. On investment grade the effect is roughly half as large, and for Black--Cox and the subordinator it is absent. The table shows the cohort board on high-yield; we omit the full basket because its growth rates agree with the cohort's to within one percentage point for every model.

\begin{table}[t]
\centering\small
\begin{tabular}{lrrrrrr}
\toprule
 & \multicolumn{3}{c}{NAHY} & \multicolumn{3}{c}{NAIG} \\
\cmidrule(lr){2-4}\cmidrule(lr){5-7}
Model & 0--30d & 31--60d & 61--92d & 0--30d & 31--60d & 61--92d \\
\midrule
Black--Cox & 0.0128 & 0.0138 & 0.0142 & 0.0031 & 0.0031 & 0.0031 \\
CIR & 0.0127 & 0.0144 & 0.0156 & 0.0025 & 0.0026 & 0.0028 \\
JCIR & \textbf{0.0090} & 0.0103 & 0.0113 & 0.0018 & 0.0019 & 0.0020 \\
subordinator & 0.0141 & 0.0155 & 0.0163 & 0.0052 & 0.0052 & 0.0052 \\
CL & 0.0099 & 0.0110 & 0.0116 & 0.0019 & 0.0020 & 0.0020 \\
LBM+Exp & \textbf{0.0090} & \textbf{0.0100} & \textbf{0.0108} & \textbf{0.0017} & \textbf{0.0018} & \textbf{0.0019} \\
\bottomrule
\end{tabular}
\caption{Static-parameter staleness. Median out-of-sample RMSE by age of the statics (days since the fit window closed), pooled across rolling folds.}
\label{tab:rolling-staleness}
\end{table}

The second question is how much history a fit should use. We refit the statics of all six models on trailing windows of $3$, $6$, $9$, and $12$ months, each ending 2022-05-27, and score every fit on the common block 2022-05-28 to 2023-06-30 with only the state re-fitted; Table~\ref{tab:window-sweep} reports the median out-of-sample RMSE, with the lowest entry in each row bolded. The three-month window is lowest, or tied lowest, for ten of the twelve model-index pairs. The two exceptions are JCIR on high-yield, where six months is lower by one unit in the fourth decimal place, and the subordinator on investment grade, where three months is marginally the worst of the four. Because the windows share an end date, the ordering does not reflect recency but the age of the data inside the fit: the longer windows reach back into 2021, when Figure~\ref{fig:fold-series} shows error at roughly half its 2022 level.

\begin{table}[t]
\centering\small
\begin{tabular}{lcccccccc}
\toprule
 & \multicolumn{4}{c}{NAHY} & \multicolumn{4}{c}{NAIG} \\
\cmidrule(lr){2-5}\cmidrule(lr){6-9}
Model & 3m & 6m & 9m & 12m & 3m & 6m & 9m & 12m \\
\midrule
Black--Cox & \textbf{0.0189} & 0.0191 & 0.0194 & 0.0198 & \textbf{0.0039} & \textbf{0.0039} & \textbf{0.0039} & \textbf{0.0039} \\
CIR & \textbf{0.0198} & 0.0208 & 0.0223 & 0.0230 & \textbf{0.0029} & 0.0033 & 0.0040 & 0.0044 \\
JCIR & 0.0143 & \textbf{0.0142} & 0.0148 & 0.0149 & \textbf{0.0023} & 0.0024 & 0.0026 & 0.0030 \\
subordinator & \textbf{0.0195} & 0.0207 & 0.0212 & 0.0218 & 0.0058 & \textbf{0.0055} & \textbf{0.0055} & 0.0056 \\
CL & \textbf{0.0150} & 0.0154 & 0.0155 & 0.0160 & \textbf{0.0021} & \textbf{0.0021} & 0.0022 & 0.0022 \\
LBM+Exp & \textbf{0.0142} & 0.0145 & 0.0147 & 0.0150 & \textbf{0.0021} & \textbf{0.0021} & \textbf{0.0021} & 0.0022 \\
\bottomrule
\end{tabular}
\caption{Window-length sweep. Median out-of-sample RMSE on the common block 2022-05-28 to 2023-06-30, statics fitted on trailing windows of 3--12 months ending 2022-05-27. The lowest entry in each row is bolded.}
\label{tab:window-sweep}
\end{table}

The two tables bear on the same design choice from opposite sides. Shortening the window from twelve months to three lowers high-yield out-of-sample error by between four and fourteen per cent, which is smaller than the within-fold decay of Table~\ref{tab:rolling-staleness}. 
We therefore adopt a trailing three-month window, re-fitted quarterly, and warm-started from the preceding fold.

\subsection{The Spread-Inversion Test}\label{sec:inversion}

A 1y--5y inversion is a firm-date at which the quoted curve satisfies $S_{\mathrm{par}}(1\mathrm{y})>S_{\mathrm{par}}(5\mathrm{y})$. Subsection~\ref{sec:sigma-zero} sets up a falsifiable prediction, which we test here on the data. For each model we report recall, the fraction of observed inverted firm-dates at which the fitted par-spread curve also inverts, with false positives counted on an equal-size control panel of CCC firm-dates whose quoted curves are not inverted. For NAHY, there are $170$ inversions in the first year, across five firms, and $212$ in the second, across seven; over $95\%$ carry CCC ratings, and each year is dominated by names that subsequently filed for bankruptcy. For NAIG, there are none in either of the first two years. In Table~\ref{tab:inversion-recall}, the in-sample column uses the year-one fits and the out-of-sample column freezes the year-one statics and re-fits only the per-date state on year-two. The result for the subordinator is the one the lemma asserts: $0$ of $170$ in sample and $0$ of $212$ out of sample, with no false positives either.

The inversion test divides the rest of the field. LBM+Exp retains $197$ of $212$ out of sample, some $93\%$, because the state $\delta$ doubles as a shape parameter: taking it small front-loads the hazard while keeping the limiting behaviour of the driver unchanged. The affine intensities go the other way, JCIR to $2$ of $212$ and CIR to $33$ of $212$ (see Remark~\ref{rmk:cirhaz}). Black--Cox retains $185$ of $212$, second only to LBM+Exp, despite sitting in the bottom half of the field by RMSE: the first-passage density of a Brownian motion started close to its barrier is unimodal, so a small fitted $x_0$ gives a hazard that rises and then falls. 

Cram\'er--Lundberg retains $177$ of $212$ out of sample, or $84\%$: below LBM+Exp and Black--Cox, far above the affine intensities, and consistent with a driver whose state is also a shape parameter. Its shortfall is not spread across those firm-dates but concentrated in two of the seven firms, and one of them is instructive. On Rite Aid, which filed for bankruptcy a few months after the second year closed, the $\mu\le0$ constrained fit lands at $\mu=-8\times10^{-9}$, which is to say on the $\mu=0$ boundary at which Lemma~\ref{lem:inversion} forbids inversion, and the model recalls none of that firm's $26$ inverted firm-dates; at any strictly negative drift, $-0.05$ included, it recovers all $26$.

\begin{table}[t]
\centering\small
\begin{tabular}{lrrrr}
\toprule
 & \multicolumn{2}{c}{In-sample (Y1)} & \multicolumn{2}{c}{Out-of-sample (Y2)} \\
\cmidrule(lr){2-3}\cmidrule(lr){4-5}
Model & Recall & FP & Recall & FP \\
\midrule
Black--Cox & 163/170 & 23 & 185/212 & 26 \\
CIR & 149/170 & 23 & 33/212 & 59 \\
JCIR & 131/170 & 7 & 2/212 & 5 \\
subordinator & 0/170 & 0 & 0/212 & 0 \\
CL & 158/170 & 20 & 177/212 & 49 \\
LBM+Exp & 162/170 & 15 & 197/212 & 32 \\
\bottomrule
\end{tabular}
\caption{Spread-curve inversion test on CDX-NAHY; fraction of observed-inverted firm-dates ($S_{\mathrm{par}}(1\mathrm{y})>S_{\mathrm{par}}(5\mathrm{y})$) where the model's par-spread curve also inverts, with false positives on an equal-size CCC control. Out-of-sample freezes year-one statics and refits only the per-date state. CDX-NAIG has zero inverted firm-dates in either year.}
\label{tab:inversion-recall}
\end{table}

\paragraph{Which sign of $\mu$ the data prefers.}
Two independent pieces of evidence favour the negative-drift class. First, when \eqref{eq:driver} is fitted with $\sigma$ held at zero and $\mu$ free, $94$ of the $100$ high-yield firms land on the Cram\'er--Lundberg side, $\mu<0$. Adding back the Brownian component at free drift lowers median panel error from $0.0098$ to $0.0090$, while forcing $\sigma=0$ and $\mu\ge0$ raises it to $0.0163$. Second, when the fitted curves are extrapolated beyond the quoted pillars, none of $400$ sampled subordinator fits produces a decreasing forward hazard anywhere on the 5y--30y segment, exactly as the plateau of Lemma~\ref{lem:plateau} predicts, against $35\%$ of Cram\'er--Lundberg fits that roll over after 5y.

\section{Joint Calibration and Index Tranches}\label{sec:joint-calibration}

\subsection{Tranche Calibration Target and Independence Baseline}\label{sec:joint-independence}

The single-name calibration of Section~\ref{sec:calibration} provides a marginal law for $\tau_i$, for each firm and date. This section is about the joint law of $(\tau_1,\dots,\tau_n)$, and we start by examining which instruments carry information about it.

The vanilla index swap does not. It pays protection on each constituent separately, so both of its legs are sums over firms of quantities depending on that firm's marginal law alone, and, by linearity, any change of dependence structure that preserves the marginals leaves the index spread unchanged. A tranche behaves differently. It absorbs only the portion of the portfolio loss falling in its attachment band, which is a nonlinear function of the number of defaults, and it is for that reason sensitive to whether defaults arrive together or apart. The tranches are thus suitable instruments to test the joint law, and they lead to our calibration target. 

Appendix~\ref{sec:tranche-calibration} converts a default-count law into the four tranche \emph{upfronts} and defines the \emph{tranche error} as model upfront minus observed upfront, which is the calibration target. Note that tranche upfronts are given as fractions of notional and take values in $[-1,1]$. We score a model by summing the absolute tranche error over the four tranches. 
Write $p_i(t)$ for firm $i$'s model-implied default probability at the on-the-run 5Y horizon and $N_T=\sum_i\ind\{\tau_i\le T\}$ for the default count. Under independent marginals, $N_T$ is a sum of independent (inhomogeneous) Bernoulli random variables. Priced from independent marginals, the tranche panel carries an error of $0.40$ to $0.44$ per board (Table~\ref{tab:rolling-tranche}). 

Every joint model below is scored using the rolling protocol of Subsection~\ref{subsec:panel}, applied to the tranches. Within each on-the-run series, all statics and dependence parameters are fitted on three months of quotes and the model then marks the following quarter under three levels of re-marking. With dependence \emph{frozen}, only the per-firm states $\delta_i$ are re-anchored to each day's CDS curves, so the tranche quotes supply nothing and the column is a genuine prediction. Hereafter \emph{frozen} always refers to this marking level for the dependence parameters. With the \emph{frequency} re-marked, the model is additionally allowed one parameter, the systemic arrival rate, fitted to that day's tranche quotes; with dependence \emph{fully} re-marked, the entire dependence vector is refitted. The levels cost different numbers of parameters in different models: two for the common-factor constructions of Section~\ref{sec:common-wh}, six for Duffie--G\^arleanu, and one for the single-factor Gaussian copula approach of \cite{li_default_2000}, whose single correlation makes its only available re-marking a full one. 

The tranches we price in this section are 5Y contracts, so no part of a single-name curve beyond the five-year point enters a tranche price. All single-name marking within the joint calibration therefore uses only the pillars up to and including 5y, which spares the joint fit from continuing out to thirty years for information the target does not contain.

\subsection{Bilevel Calibration}\label{sec:joint-bilevel-cal}

We carry out a bilevel calibration for each of the subcases of \eqref{eq:driver} plus a common compound Poisson $Y$ with rate $\lambda_c$ and exponential jumps with mean $\gamma$ (the joint model of Subsection~\ref{sec:joint-construction}). The outer loop searches the two dependence scalars $(\gamma,\lambda_c)$, minimizing the summed absolute tranche error over the frozen quarter with a penalty protecting single-name fit quality. For the inner loop, the dependence parameters $(\gamma,\lambda_c)$ are held fixed, and the per-firm static parameters together with the per-firm, per-date state $\delta_i$ are fitted against the observed hazard curves under \eqref{eq:common-marginal}. The state solve is cheap: the default probability is monotone in $\delta_i$, so the solve is a bracketed one-dimensional root-find. Tranche prices are computed by the Wiener--Hopf simulator under common random numbers. On the marking ladder of Subsection~\ref{sec:joint-independence}, the frequency level re-marks $\lambda_c$ and the full level re-marks both $\lambda_c$ and $\gamma$.

Because the three drivers differ only in the idiosyncratic component, comparing them isolates the driver class at fixed dependence. This gives a joint-level counterpart of the single-name comparison of Section~\ref{sec:calibration}. For the Cram\'er--Lundberg and LBM+Exp drivers the median fitted $\lambda_c$ is $0.15$ to $0.16$ per year on high-yield, a systemic event every six or seven years, with median severity $\gamma$ of $3.1$ to $4.1$ against median fitted distance-to-distress levels of roughly $10$ to $11$: each arrival erodes about a third of the median distance-to-distress. On investment grade the median $\lambda_c$ is higher and driver-dependent: $0.25$ for Cram\'er--Lundberg with severity $2.3$ against a median distance-to-distress of $12$, and $0.43$ for LBM+Exp with severity $1.5$ against a median distance-to-distress of $9$.
The median fitted product $\lambda_c\gamma$ across folds is $0.52$ per year for LBM+Exp and $0.60$ for Cram\'er--Lundberg on high-yield, and $0.61$ and $0.62$ respectively on investment grade, so on both indices the common factor erodes five to seven per cent of the median firm's distance-to-distress each year. The comparable erosion rates, and the higher fitted $\lambda_c$ on investment grade, reverse the direction one might expect from physical default rates, but they are natural under the risk-neutral measure, where $\lambda_c$ prices the market's weight on a systemic event rather than its physical frequency. A recession or financial crisis strikes both universes at the same calendar time, so the common term prices the same tail event on either index, and cross-index differences are left to the idiosyncratic drivers. Senior investment-grade tranches behave like disaster insurance, carrying near-zero expected loss in most states and a wipeout in a systemic event. Pricing them requires a substantial risk-neutral weight on that tail. The subordinator sits apart in these coordinates, fitting $\lambda_c\approx0.08$ with severity $0.7$ on high-yield and $\lambda_c\approx0.34$ with severity below $0.1$ on investment grade, so that its fitted product is an order of magnitude smaller than the negative-drift drivers', $0.06$ and $0.03$ respectively. Since the subordinator class lacks downward drift, anything higher would wipe out massive proportions of the indices.

\subsection{Benchmarks and Results}\label{sec:joint-results}

We compare two benchmarks to our construction from Subsection~\ref{sec:joint-construction}. The first is a single-factor Gaussian copula approach (see \cite{li_default_2000}) applied to the market marginals directly with a single flat correlation $\rho$ fitted to the year-one tranche panel. The copula correlates the default times through a Gaussian factor: firm $i$ defaults by horizon $T$ if and only if
\[
\sqrt{\rho}\,Z+\sqrt{1-\rho}\,\epsilon_i\;\le\;\Phi^{-1}\bigl(p_i(T)\bigr),
\]
with $Z,\epsilon_1,\dots,\epsilon_n$ i.i.d.\ standard normal and $p_i(T)$ now taken to be the bootstrapped market default probability rather than a model\textquotesingle s, so that, conditional on $Z$, the defaults are independent and the count law follows by deterministic Gauss--Hermite quadrature over $Z$. Because it takes the bootstrapped curves as given, its single-name error is zero by construction, which isolates the quality of the dependence structure from the quality of the marginals. Market practice improves the flat correlation with the base-correlation skew, one correlation per attachment point (this is even published in the data, as the \texttt{BaseCorrelation} field, and see also \cite{sp_global_market_intelligence_credit_2024}). The skew gives four separate models rather than a joint law, so we take the flat-$\rho$ copula as the like-for-like benchmark. Our second benchmark is the affine model of \cite{duffie_risk_2001}, in which each firm's default intensity is the sum of an idiosyncratic and a common JCIR process: twelve parameters in all, six of them common. We price this semi-analytically, with Monte Carlo in the one-dimensional common factor and the idiosyncratic parts in closed form.

\begin{table}[t]
\centering
\small
\begin{tabular}{lrrrrrr}
\toprule
 & \multicolumn{3}{c}{$\sum|\text{tranche error}|$} & \multicolumn{3}{c}{single-name RMSE} \\
\cmidrule(lr){2-4}\cmidrule(lr){5-7}
Model & frozen & frequency & full & frozen & frequency & full \\
\midrule
\multicolumn{7}{l}{\textbf{NAHY cohort}} \\
Independence & 0.398 & --- & --- & 0.0045 & --- & --- \\
subordinator & 0.127 & 0.091 & 0.074 & 0.0113 & 0.0116 & 0.0111 \\
CL & 0.113 & 0.084 & 0.057 & 0.0069 & 0.0069 & 0.0067 \\
LBM+Exp & 0.101 & 0.079 & 0.056 & 0.0045 & 0.0047 & 0.0052 \\
Duffie--G\^arleanu & 0.247 & 0.230 & 0.212 & 0.0207 & 0.0182 & 0.0176 \\
Gaussian copula & 0.192 & --- & 0.159 & --- & --- & --- \\
\addlinespace
\multicolumn{7}{l}{\textbf{NAHY full basket}} \\
Independence & 0.441 & --- & --- & 0.0050 & --- & --- \\
subordinator & 0.121 & 0.092 & 0.072 & 0.0120 & 0.0116 & 0.0115 \\
CL & 0.117 & 0.088 & 0.060 & 0.0078 & 0.0074 & 0.0079 \\
LBM+Exp & 0.118 & 0.093 & 0.066 & 0.0052 & 0.0053 & 0.0054 \\
Duffie--G\^arleanu & 0.269 & 0.251 & 0.231 & 0.0199 & 0.0187 & 0.0187 \\
Gaussian copula & 0.226 & --- & 0.189 & --- & --- & --- \\
\addlinespace
\multicolumn{7}{l}{\textbf{NAIG}} \\
Independence & 0.421 & --- & --- & 0.0008 & --- & --- \\
subordinator & 0.447 & 0.442 & 0.347 & 0.0093 & 0.0093 & 0.0060 \\
CL & 0.053 & 0.036 & 0.026 & 0.0011 & 0.0013 & 0.0012 \\
LBM+Exp & 0.048 & 0.031 & 0.022 & 0.0010 & 0.0011 & 0.0011 \\
Duffie--G\^arleanu & 0.392 & 0.385 & 0.380 & 0.0060 & 0.0064 & 0.0054 \\
Gaussian copula & 0.105 & --- & 0.097 & --- & --- & --- \\
\bottomrule
\end{tabular}
\caption{Tranche rolling folds (10 per board, series-aligned, 3m fit / 3m test). Mean over folds of the per-fold mean of the total absolute tranche error by marking level, and the median single-name hazard RMSE under each level. The copula prices market marginals exactly, so its single-name column is shown as ---. The three common-jump models share the dependence structure and differ only in the idiosyncratic driver.}
\label{tab:rolling-tranche}
\end{table}

The results of the rolling window experiment are given in Table~\ref{tab:rolling-tranche}. Folds are aligned with the annual series rolls: within each on-the-run series the dependence parameters and statics are fitted on three months and marked on the following three, giving ten folds per board. In the frozen column, where the dependence parameters are not re-marked, the common-jump construction closes most of the independence gap: the summed error falls from $0.398$ to $0.101$ on the high-yield cohort ($75\%$) and from $0.421$ to $0.048$ on investment grade ($89\%$), with LBM+Exp leading and Cram\'er--Lundberg within about a hundredth of it everywhere, marginally ahead on the full basket. Allowing the full dependence vector to be re-marked closes the gap further, to $0.056$ and $0.022$ ($86\%$ and $95\%$). The consistency of the single-name RMSE columns shows that our dependence construction does not corrupt the single-name fits.

The benchmarks trail on every board, and re-marking their dependence parameters to each day's tranche quotes closes little of the remaining gap. The copula's frozen error (cohort $0.192$ / basket $0.226$ / IG $0.105$) falls only to $0.159$ / $0.189$ / $0.097$ once its single correlation is re-marked. The Duffie--G\^arleanu benchmark does worse still: frozen at $0.247$ / $0.269$ / $0.392$, it is barely improved by re-marking all six of its common-factor parameters, and it carries single-name errors three to six times those of the running-supremum class. On investment grade in particular, the Duffie--G\^arleanu benchmark recovers almost no dependence ($0.380$ fully re-marked against an independence level of $0.421$), and its fitted common factor sits at the boundary of its admissible range, no mean reversion and maximal volatility, which we read as misspecification of the six-parameter common-factor form rather than under-searching.

The subordinator model is essentially tied with the other two drivers on high-yield under the frozen protocol, within $0.03$ of both on the cohort ($0.127$ against $0.101$ and $0.113$) and the full basket ($0.121$ against $0.117$ and $0.118$), because our rolling marks re-anchor on the $\le$5y segment, over which rising hazards are within reach of a monotone driver and the plateau effect from Lemma~\ref{lem:plateau} is not as binding. On investment-grade, by contrast, its quarterly fits are unstable: they land in small-severity corners in sample and mark worse than independence out of sample ($0.447$ frozen against $0.421$). The case against the $\mu\ge0$ driver therefore rests on the single-name field of Section~\ref{sec:calibration}, the inversion test of Subsection~\ref{sec:inversion}, and the investment-grade tranches, not on the high-yield tranche results alone.

\paragraph{Numerical controls.}

None of the results above rests on uncontrolled simulation error, and most of the pricing is not Monte Carlo at all. The copula approach and the subordinator marginals are deterministic; every Monte Carlo construction runs under common random numbers, so ladders are bit-reproducible; independent seed streams shift the monthly-mean summed error by between $0.010$ and $0.024$ mean absolute depending on the driver, an order of magnitude below the separations from independence and from the benchmarks in Table~\ref{tab:rolling-tranche}; and refining the Wiener--Hopf step rate from $q=16$ to $q=48$ moves individual tranche upfronts within that same Monte Carlo noise while preserving every ordering, so the rate $q=16$ is adopted throughout.

\section{Concluding Remarks}\label{sec:conclusion}
We have developed a multi-credit model that accounts for simultaneous defaults and is simulated with no time-discretization bias by a Wiener--Hopf Monte Carlo scheme. On daily CDX tranche data, the resulting two-parameter dependence structure closes $73\%$ to $89\%$ of the pricing gap left by independent marginals with dependence frozen, up to $95\%$ with dependence parameters fully re-marked, and dominates benchmark models including the one-factor Gaussian copula and the affine intensity model of Duffie--G\^arleanu. Because each firm's marginal remains phase-type under the common shock, the same root system that prices one firm's CDS curve is reused in pricing the whole portfolio's joint law, and the single-name comparison of Section~\ref{sec:calibration} shows that this building block is sound.

The multi-credit construction applies broadly to tranched products, whose payoffs are nonlinear functions of the number and timing of defaults within a pool. Beyond the corporate-credit setting considered here, it can be used for credit risk transfer (CRT) bonds, through which the government-sponsored enterprises pass the credit risk of a pool of residential mortgages to private investors in tranched form (see \cite{capponi_pricing_2026}). Therein, a common systemic shock correlating simultaneous mortgage defaults across the pool is the same object as $Y$ in \eqref{eq:joint-common}, so the dependence layer of our construction carries over directly; what must change is the marginal driver, because mortgage default risk is governed principally by the risk-neutral evolution of the borrower's income and of the house price relative to the loan balance, rather than by a firm's asset value or credit spread. Another application is the broader class of collateralized debt obligations (see \cite{brigo_credit_2010} for a survey), where the common-jump construction is a direct two-parameter alternative to the base-correlation approach, fitting the whole tranche stack from one joint law and one exact simulation scheme rather than one correlation per attachment point. We leave both applications for future work.

\appendix

\section{Laplace Inversion with a Talbot Contour}\label{sec:Talbot}
The default probability $\bQ(\tau_i\le T)$ at maturity $T$ is recovered from the Laplace transform \eqref{eq:thm} by numerical inversion along a Talbot contour (see \cite{abate_unified_2006}). The contour is given by
\[
q(\theta) \;=\; r\,\theta\,(\cot\theta + i),\qquad \theta\in(0,\pi),\qquad r := \frac{2N_\theta}{5T},
\]
which wraps around the negative real axis. With $\theta_l = l\pi/N_\theta$ for $l = 0,1,\dots,N_\theta-1$, weights $w_0=\tfrac{1}{2}$ and $w_l = 1 + i\lb(\theta_l + (\theta_l\cot\theta_l - 1)\cot\theta_l\rb)$ for $l\ge 1$, and the convention $q(\theta_0):=r$, the trapezoidal quadrature on this contour reads
\[
\bQ(\tau_i\le T)
\;\approx\;
\frac{r}{N_\theta}\,\real\!\lb[\,\sum_{l=0}^{N_\theta-1} w_l\,e^{q(\theta_l) T}\,\hat F\!\lb(q(\theta_l);\,\delta\rb)\,\rb],
\]
where $\hat F(q;\delta)$ is the closed-form Laplace transform from \eqref{eq:thm}. Convergence in $N_\theta$ is geometric (see \cite{abate_unified_2006}), and we take $N_\theta=18$ throughout. The deformation onto the Talbot contour requires the singularities of the continued transform (Remark~\ref{rmk:distinct}) to lie to its left; we audit the fitted parameter ranges numerically to confirm that the contour and its nodes stay clear of the small-root coalescence locus and of the points where some $\beta_j(q)\in\{0,1\}$. A computational observation important for the panel study in Section~\ref{sec:calibration}: the Talbot nodes $q(\theta_l)$, the polynomial roots $\beta_j(q(\theta_l))$, the residue coefficients $A_j(q(\theta_l))$, and the weights $w_l$ all depend only on the maturity $T$ and the static parameters, not on the per-date state $\delta$. Once these are cached per (firm, tenor)-pair, each per-date evaluation reduces to a vector of complex exponentials $e^{\beta_j\delta_t}$ summed over $j$ and $l$, with no further root-finding.

\section{Tranche Pricing Conventions and the Joint-Default Calibration Target}\label{sec:tranche-calibration}

This appendix collects the tranche-pricing identities that convert observed index tranche quotes into a calibration target for Section~\ref{sec:joint-independence}. 

Fix an index and a time horizon $T$. Let $n$ be the number of constituents and $R$ be the fixed recovery (Section~\ref{sec:data}). For $s\in[0,T]$, we let $N_s$ denote the cumulative default count, and $L_s = (1-R)N_s/n$ be the portfolio loss fraction. A joint default model $m$ gives a family of probability mass functions $f_k^m(s) := \bQ^m(N_s = k)$, $s\in[0,T]$. For a tranche with attachment $a$ and detachment $b$, the cumulative tranche loss is $L^{a,b}_s = \min(L_s,b)-\min(L_s,a)$ and the remaining notional is $O_s = (b-a) - L^{a,b}_s$. The protection leg pays tranche losses as they occur; the premium leg pays the index's fixed running coupon $S^\star$ (Table~\ref{tab:data-tranche}) on the remaining notional plus accrual-on-default. With $L_k := (1-R)k/n$, both legs of the tranche are linear functionals of $f^m$,
\begin{equation}\label{eq:tranche-el-ref}
\E^m\!\lb[L^{a,b}_s\rb] \;=\; \sum_{k=0}^{n} f_k^m(s)\cdot \min\,\bigl(\max(L_k - a,\,0),\, b - a\bigr),
\end{equation}
and $\E^m\lb[O_s\rb] = (b-a) - \E^m\lb[L^{a,b}_s\rb]$. Explicitly, with $D(0,\cdot)$ the discount curve of Section~\ref{sec:data} and $\{t_p\}$ the quarterly premium dates with year-fractions $\zeta_p$,
\[
\mathrm{Prot}^m = \int_0^T D(0,s)\,d\E^m\lb[L^{a,b}_s\rb],
\quad
\mathrm{RPV01}^m = \sum_{t_p\le T} D(0,t_p)\,\zeta_p\,\E^m\lb[O_{t_p}\rb] + \text{accrual},
\]
both evaluated on the time grid carrying $f^m$. Both indices trade against the fixed coupon with the balance settled at inception, and we work with the \emph{upfront} fraction $U\in[-1,1]$ of tranche notional; the upfront implied by model $m$ is $U^m([a,b]) = (\mathrm{Prot}^m - S^\star\,\mathrm{RPV01}^m)/(b-a)$, and the empirical \emph{tranche error} is $\Delta^m([a,b],t) = U^m - U^{\mathrm{obs}}$, the observable against which the dependence parameters can be calibrated. The default-count thresholds at which successive tranches begin absorbing losses (the smallest $k$ with positive tranche loss) are $\lfloor n\,a/(1-R)\rfloor + 1$: $\{1,22,36,51\}$ on the high-yield index and $\{1,7,15,32\}$ on the investment-grade index.

\section{ISDA Bootstrap of Forward Hazard Rates}\label{sec:isda-bootstrap}

The calibration target of Section~\ref{sec:calibration} is a vector of bootstrapped forward hazard rates, obtained from the quoted CDS par spreads by a standard ISDA-style bootstrap; see \cite[Subsection 3.1.4]{brigo_counterparty_2013} for the underlying CDS pricing identities and their inversion to default probabilities. Fix a firm and a date with quoted maturities $T_1 < \cdots < T_M$, par spreads $S(T_m)$, recovery $R$ (Markit field \texttt{cdsassumedrecovery}; this per-firm $R$ is distinct from the fixed index recovery of Appendix~\ref{sec:tranche-calibration}), and the bootstrapped risk-free discount factors $D(0,s)$. Model the survival function as $\bar F(s) = \exp(-\int_0^s h(u)\,du)$ with $h$ constant on each $(T_{m-1}, T_m]$. For a contract with quarterly premium dates $\{t_p\}$ and accrual-on-default, the protection and premium legs are
\begin{align*}
\mathrm{Prot}(T_m) &= (1-R)\int_0^{T_m} D(0,s)\,(-d\bar F(s)), \\
\mathrm{Prem}(T_m) &= S(T_m)\sum_{t_p \le T_m} D(0,t_p)\,\zeta_p\,\bar F(t_p) + \text{(accrual-on-default)},
\end{align*}
with $\zeta_p$ the year-fraction of premium period $p$ and the accrual term integrating, over each period, $S(T_m)\,D(0,s)\,(s-t_{p-1})\,(-d\bar F(s))$. The forward hazard on $(T_{m-1},T_m]$ is solved sequentially ($m=1,\dots,M$) from the par condition $\mathrm{Prot}(T_m) = \mathrm{Prem}(T_m)$, holding the previously solved segments fixed: a one-dimensional root-find per pillar. The resulting hazard vector $h^{\mathrm{boot}}(T_m)$ is the target against which each model's implied forward hazard $h^\theta(T_m)$ (from $\bQ(\tau\le T)$) is compared in the RMSE loss of Section~\ref{sec:calibration}. 

\section*{Acknowledgments}
Special thanks to Ben Hambly, Philip Protter, Max Reppen, and Andreas S{\o}jmark for sharing their insights and for many helpful discussions. Additionally, we thank the organizers and participants of the XIII Bachelier World Congress (Bologna, Italy) for the opportunity to present this work, and for useful feedback.

\printbibliography[heading=bibintoc]

\end{document}